\documentclass[11pt,letterpaper]{article}

\usepackage{amsmath,amsfonts}
\usepackage{amssymb}
\usepackage{amsthm}
\usepackage{graphicx}
\usepackage{verbatim}
\usepackage{mathrsfs}
\usepackage{algorithm,algorithmic}
\usepackage{url}
\usepackage{ogonek}
\usepackage{enumerate}

\usepackage{times}
\usepackage{fullpage}
\newtheorem{lemma}{Lemma}
\newtheorem{thm}{Theorem}
\newtheorem{cor}{Corollary}
\newtheorem{defn}{Definition}
\newtheorem{obs}{Observation}
\newtheorem{claim}{Claim}

\newcommand{\E}{\mathrm{E}}
\newcommand{\st}{\mbox{$s$-$t$} }

\algsetup{
	linenodelimiter=\mbox{ }
}

\title{Improving Christofides' Algorithm for the $s$-$t$ Path TSP}

\author{ 
Hyung-Chan An
\thanks{
{\tt anhc@cs.cornell.edu}.
Dept. of Computer Science, Cornell University, Ithaca, NY 14853.
Research supported in part by NSF under grants no. CCF-1017688 and CCF-0729102, and the Korea Foundation for Advanced Studies.
Part of this research was conducted while the author was a visiting student at CSAIL, MIT.
}
\and
Robert Kleinberg
\thanks{
{\tt rdk@cs.cornell.edu}.
Dept. of Computer Science, Cornell University, Ithaca, NY 14853.
Supported by NSF grants CCF-0643934 and CCF-0729102, AFOSR grant FA9550-09-1-0100, a Microsoft Research New
Faculty Fellowship, a Google Research Grant, and an Alfred P. Sloan Foundation Fellowship.
}
\and
David B. Shmoys
\thanks{
{\tt shmoys@cs.cornell.edu}.
School of ORIE and Dept. of Computer Science, Cornell University, Ithaca, NY 14853.
Research supported in part by NSF under grants no. CCF-0832782 and CCF-1017688.
Part of this research was conducted while the author was a visiting professor at Sloan School of Management, MIT.
}
}

\date{}

\begin{document}

\maketitle 
\def\thepage {} 
\thispagestyle{empty}

\begin{abstract}
We present a deterministic $\left(\frac{1+\sqrt{5}}{2}\right)$-approximation algorithm for the \st path TSP for an arbitrary metric. Given a symmetric metric cost on $n$ vertices including two prespecified endpoints, the problem is to find a shortest Hamiltonian path between the two endpoints; Hoogeveen showed that the natural variant of Christofides' algorithm is a $5/3$-approximation algorithm for this problem, and this asymptotically tight bound in fact has been the best approximation ratio known until now. We modify this algorithm so that it chooses the initial spanning tree based on an optimal solution to the Held-Karp relaxation rather than a minimum spanning tree; we prove this simple but crucial modification leads to an improved approximation ratio, surpassing the 20-year-old barrier set by the natural Christofides' algorithm variant. Our algorithm also proves an upper bound of $\frac{1+\sqrt{5}}{2}$ on the integrality gap of the path-variant Held-Karp relaxation. The techniques devised in this paper can be applied to other optimization problems as well: these applications include improved approximation algorithms and improved LP integrality gap upper bounds for the prize-collecting \st path problem and the unit-weight graphical metric \st path TSP.
\end{abstract}
\newpage
\pagenumbering {arabic}
\normalsize

\section{Introduction}

After 35 years, Christofides' $3/2$-approximation algorithm~\cite{C} still provides the best performance guarantee known for the metric traveling salesman problem (TSP), and improving upon this bound is a fundamental open question in combinatorial optimization. For the path variant of the metric TSP in which the aim is to find a shortest Hamiltonian path between given endpoints $s$ and $t$, Hoogeveen~\cite{H} showed that the natural variant of Christofides' algorithm yields an approximation ratio of $5/3$ that is asymptotically tight, and this has been the best approximation algorithm known for this \st path variant for the past 20 years. Recently, there has been progress for the special case of metrics derived as shortest paths in unit-weight (undirected) graphs: Oveis Gharan, Saberi, and Singh~\cite{OSS} gave a $(3/2-\epsilon_0)$-approximation algorithm for the TSP, which can be extended to yield an analogous result of a $(5/3-\epsilon_1)$-approximation algorithm for the \st path TSP in the same special case (see Appendix~\ref{ap:c53}). M\"omke and Svensson~\cite{MS} gave a $1.4605$-approximation algorithm for the same special case of the TSP, as well as a $1.5858$-approximation algorithm for the \st path TSP in the same case (where the results of Appendix~\ref{ap:c53} and M\"omke \& Svensson~\cite{MS} were obtained independently). We note the techniques devised in these results for the unit-weight graphical metric case proved useful in both path and ordinary (circuit) variants. The main result of this paper is to provide the first improvement for the \emph{general} metric case of the \st path TSP: more specifically, we give a deterministic $\left(\frac{1+\sqrt{5}}{2}\right)$-approximation algorithm for the metric \st path TSP for an arbitrary metric, breaking the $5/3$ barrier. It remains an open question whether these techniques can be extended to yield a comparable improvement (over the $3/2$ barrier) for the general-metric ordinary (circuit) TSP.

Our analysis gives the first constant upper bound on the integrality gap of the path-variant Held-Karp relaxation as well, showing it to be at most the golden ratio, $\frac{1+\sqrt{5}}{2}$. We will also demonstrate how the techniques devised in the present paper can be applied to other problems, such as the prize-collecting \st path problem and the unit-weight graphical metric \st path TSP, to obtain better approximation ratios and better LP integrality gap upper bounds than the current best known.

Proposed by Held and Karp \cite{HK} originally for the circuit problem, the Held-Karp relaxation~\cite{HK} is a standard LP relaxation to (the variants of) TSP, and has been successfully used by many algorithms \cite{BGSW, G:pc, AGMOS, AKS, OSS, MS, M}. In the LP-based design of an approximation algorithm, one important measure of the strength of a particular LP relaxation is its integrality gap, i.e., the worst-case ratio between the integral and fractional optimal values; however, there exists a significant gap between currently known lower and upper bounds on the integrality gap of the Held-Karp relaxation. For the circuit case, the best upper bound known of $3/2$ is constructively proven by the analyses of Christofides' algorithm due to Wolsey~\cite{W} and Shmoys \& Williamson~\cite{SW}; yet, the best lower bound known is $4/3$, achieved by the family of graphs depicted in Figure~\ref{f:ig}(a) under the unit-weight graphical metric~\cite{G:ineq}. For the path problem, Hoogeveen~\cite{H} shows the natural variant of Christofides' algorithm is a $5/3$-approximation algorithm, but the analysis compares the output solution value to the optimal (integral) solution; therefore it is unclear whether the algorithm yields an integrality gap bound for the Held-Karp relaxation formulated for the path problem. The analysis of the present algorithm, in contrast, reveals an upper bound of $\frac{1+\sqrt{5}}{2}$ on its integrality gap, matching the approximation ratio; we also show in Appendix~\ref{ap:c53} that an LP-based analysis of Christofides' algorithm proves a weaker upper bound of $5/3$. We observe that the family of graphs in Figure~\mbox{\ref{f:ig}(b)} establishes the integrality gap lower bound of $3/2$ under the unit-weight graphical metric. Note that this lower bound is strictly greater than the known upper bound of $(3/2-\epsilon_0)$ on the integrality gap of the circuit-variant Held-Karp relaxation under the unit-weight graphical metric; this suggests that the lack of a performance guarantee known for the \st path TSP matching the $3/2$ for other TSP variants has a true structural cause.

\begin{figure}
\center
\includegraphics[width=310pt]{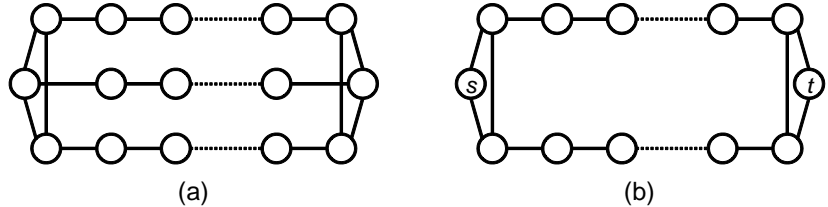}
\caption{Examples establishing the integrality gap lower bounds for the circuit- and path-variant Held-Karp relaxations.}
\label{f:ig}
\end{figure}
\begin{figure}
\center
\includegraphics[width=280pt]{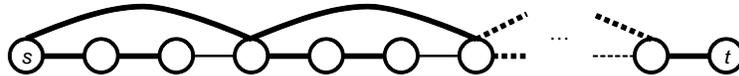}
\caption{Example showing $5/3$ is asymptotically tight~\cite{H}: a minimum spanning tree is marked with thick edges.}
\label{f:53ex}
\end{figure}

A feasible solution to the path-variant Held-Karp relaxation is in the spanning tree polytope; thus, given a feasible Held-Karp solution, there exists a probability distribution over spanning trees whose marginal edge probabilities coincide with the Held-Karp solution. The present algorithm first computes an optimal solution to the Held-Karp relaxation, and samples a spanning tree from a probability distribution whose marginal is given by the Held-Karp solution. Then it augments this tree with a minimum $T$-join, where $T$ is the set of vertices with ``wrong'' parity of degree, to obtain an Eulerian path visiting every vertex; this Eulerian path can be shortcut into an \st Hamiltonian path of no greater cost. Our analysis of this algorithm shows that the expected cost of the Eulerian path is at most $\frac{1+\sqrt{5}}{2}$ times the Held-Karp optimum; the analysis relies only on the marginal probabilities, and therefore holds for \emph{any} arbitrary distribution with the given marginals. We note that this flexibility enables a simple derandomization: a feasible Held-Karp solution can be efficiently decomposed into a convex combination of polynomially many spanning trees (see Gr{\"o}tschel, Lov{\'a}sz, and Schrijver~\cite{GLS}) and trying every spanning tree in this convex combination yields a simple deterministic algorithm. We also note that our algorithm differs from Christofides' in only one crucial respect: rather than taking a single tree and augmenting it with a $T$-join, we try out polynomially many trees and then take the one whose augmentation yields the lowest-cost path. The example in Figure~\ref{f:53ex} due to Hoogeveen~\cite{H} shows that this simple modification of the original algorithm is crucial to achieving the improved approximation ratio: if one only tries augmenting the minimum spanning tree, the approximation ratio remains no better than $5/3$.

As the expected cost of the sampled spanning tree is equal to the Held-Karp optimum, the rest of the analysis focuses on bounding the cost of the minimum $T$-join by providing a low-cost \emph{fractional $T$-join dominator} that serves as an upper bound on the cost of the minimum $T$-join. First we show that the Held-Karp solution and the spanning tree, while being costly fractional $T$-join dominators themselves, are complementary: a certain linear combination of them is a fractional $T$-join dominator whose expected cost is no greater than $2/3$ times the Held-Karp optimum, thereby recovering the same $5/3$ performance guarantee provided by Hoogeveen's analysis of Christofides' algorithm. Based on this beginning analysis, we present progressively better ways of constructing a low-cost fractional $T$-join dominator. In all of these approaches, we perturb the coefficients of the tree and the Held-Karp solution to reduce the cost of their linear combination, at the expense of potentially violating some constraints of the fractional $T$-join dominator linear program, and then we add a low-cost correction to repair the violated constraints. To construct this correction vector and to bound its cost, we show that the only potentially violated constraints correspond to \emph{narrow cuts} having a layered structure, as illustrated in Figure~\ref{f:layer}. The layered structure allows us to choose disjoint sets of representative edges for each cut and to correct the violated constraints using a sum of vectors each supported on the representative edge set of the corresponding narrow cut. We show that this idea leads to a slight improvement upon $5/3$, using the fact that the representative edge sets, while being mutually disjoint, occupy a large portion of each cut and that each narrow cut constraint has only a small probability of being violated. After that, we present a tighter analysis with a similar construction. Finally, pushing the performance guarantee towards the golden ratio requires relaxing the disjointness of the representatives to a notion of ``fractional disjointness''. We define this relaxed disjointness, construct the requisite fractionally disjoint vectors via the analysis of an auxiliary flow network, and prove the performance guarantee of $\frac{1+\sqrt{5}}{2}$. We note that neither the fractional $T$-join dominator nor the narrow cuts are actually computed by the algorithm: these progressive analyses all analyze the same single algorithm while different fractional $T$-join dominators are considered in each analysis. That is, it might be possible to obtain a better performance guarantee for the same algorithm by providing a better construction of a fractional $T$-join dominator. The narrow cuts are purely for the purpose of analysis in Section~\ref{s:improv} and never determined by the algorithm; however, their algorithmic use is explored in Section~\ref{s:appl}.

Section~\ref{s:appl} demonstrates how the present results can be applied to other problems to obtain better approximation algorithms than the current best known. We first consider the metric prize-collecting \st path problem. In a prize-collecting problem, we are given ``prize'' values defined on vertices, and the objective function becomes the sum of the ``regular'' solution cost and the total ``missed'' prize of the vertices that are not included in the solution. For example, the prize-collecting \st path problem finds a (not necessarily spanning) \st path that minimizes the sum of the path cost and the total prize of the vertices not on the path. Chaudhuri, Godfrey, Rao, and Talwar~\cite{CGRT} give a primal-dual $2$-approximation algorithm for this problem. Prize-collecting TSP, the circuit version of this problem, has been introduced in Balas~\cite{B}; Bienstock, Goemans, Simchi-Levi, and Williamson~\cite{BGSW} give a LP-rounding $2.5$-approximation algorithm, and Goemans \& Williamson~\cite{GW} show a primal-dual $2$-approximation algorithm. For both problems, Archer, Bateni, Hajiaghayi, and Karloff~\cite{ABHK} give improvement on approximation ratios: using the path-variant Christofides' algorithm as a black box, Archer et al. give a $241/121$-approximation algorithm for the prize-collecting \st path problem; a $97/49$-approximation algorithm is given for the prize-collecting TSP, using Christofides' algorithm as a black box again. For the prize-collecting (circuit) TSP, Goemans~\cite{G:pc} combines Bienstock et al.~\cite{BGSW} and Goemans \& Williamson~\cite{GW} to obtain a $1.9146$-approximation algorithm, the current best known.

As the analysis of Archer et al.~\cite{ABHK} treats Christofides' algorithm as a black box, replacing this with the present algorithm readily gives an improvement. Furthermore, we will show that, since the present analysis produces the performance guarantee in terms of the Held-Karp optimum, Goemans' analysis~\cite{G:pc} can be extended to the prize-collecting \st path problem. One obstacle is that the parsimonious property~\cite{GB} used in Bienstock et al. does not immediately apply to the path case; however, we prove that a modification to the graph and the Held-Karp solution allows us to utilize this property. This yields a $1.9535$-approximation algorithm for the prize-collecting \st path problem; the same upper bound is established on the integrality gap of the LP relaxation used.

Secondly, we study the \emph{unit-weight} graphical metric \st path TSP to present a $1.5780$-approximation algorithm. As discussed above, there has been progress for this special case in both the ordinary (circuit) TSP and the \st path TSP. In Appendix~\ref{ap:c53}, we show how the results of Oveis Gharan et al.~\cite{OSS} extend to the path case. Most recently, Mucha~\cite{M} gave an improved analysis of M\"omke \& Svensson's algorithm~\cite{MS} to prove the performance guarantee of $13/9$ for the circuit case and $19/12+\epsilon$ for the path case, for any $\epsilon>0$. We observe that the critical case of this analysis is when the Held-Karp optimum is small, and we show how to obtain an algorithm that yields a better performance guarantee on this critical case, based on the main results of this paper. In particular, we devise an algorithm that works on narrow cuts to be run in parallel with the present algorithm; this illustrates that the narrow cuts are a useful algorithmic tool as well, not only an analytic tool. Our algorithm establishes an upper bound on the integrality gap of the path-variant Held-Karp relaxation under the unit-weight graphical metric, which does not match the performance guarantee but smaller than $\frac{1+\sqrt{5}}{2}$.

\section{Preliminaries}\label{s:pre}

In this section, we introduce some definitions and notation to be used throughout this paper.

Let $G=(V,E)$ be the input complete graph with metric cost function $c:E\to\mathbb{R}_+$. \emph{Endpoints} $s,t\in V$ are given as a part of the input; we call the other vertices \emph{internal points}.

For $A,B\subset V$ such that $A\cap B=\emptyset$, $E(A,B)$ denotes the set of edges between $A$ and $B$: i.e., $E(A,B)=\{\{u,v\}\in E|u\in A,v\in B\}$. Let $E(A)$ denote the set of edges within $A$: $E(A):=\{\{u,v\}\in E|u,v\in A\}$.

For nonempty $U\subsetneq V$, let $(U,\bar U)$ denote the cut defined by $U$, and $\delta(U)$ be the edge set in the cut: $\delta(U)=E(U,\bar U)$. $(U,\bar U)$ is called an \st \emph{cut} if $|U\cap\{s,t\}|=1$; we call $(U,\bar U)$ \emph{nonseparating} otherwise.

For $x,c\in\mathbb{R}^E$ and $F\subset E$, $x(F)$ is a shorthand for $\sum_{f\in F}x_f$; $c(x)$ is $\sum_{e\in E} c_e x_e$. The incidence vector $\chi_F\in\mathbb{R}^E$ of $F\subset E$ is a $(0,1)$-vector defined as follows:\[
(\chi_F)_e := \begin{cases}
1&\textrm{if }e\in F,\\
0&\textrm{otherwise}
.\end{cases}
\]

For two vectors $a,b\in\mathbb{R}^I$, let $a\ast b\in\mathbb{R}^I$ denote the vector defined as:\[
(a\ast b)_i := a_i b_i
.\]

\begin{defn}[\cite{HK}]
\label{d:hkpath}
The \emph{path-variant Held-Karp relaxation} is defined as follows:\begin{equation}\label{e:hkpath1}
\begin{array}{lll}
\textrm{minimize}&c(x)&\\
\textrm{subject to}&x(\delta(S))\geq 1,&\forall S\subsetneq V, |\{s,t\}\cap S| = 1;\\
&x(\delta(S))\geq 2,&\forall S\subsetneq V, |\{s,t\}\cap S| \neq 1, S\neq\emptyset ;\\
&x(\delta(\{s\})) = x(\delta(\{t\})) = 1;&\\
&x(\delta(\{v\})) = 2,&\forall v\in V\setminus\{s,t\};\\
&x\geq 0.&
\end{array}
\end{equation}
\end{defn}

This linear program can be solved in polynomial time via the ellipsoid method using a min-cut algorithm to solve the separation problem \cite{GLS}. The following observation gives an equivalent formulation of \eqref{e:hkpath1}.

\begin{obs}
\label{o:hkequiv}
Following is an equivalent formulation of \eqref{e:hkpath1}:\begin{equation*}
\begin{array}{lll}
\textrm{minimize}&c(x)&\\
\textrm{subject to}&x(E(S))\leq |S|-1,&\forall S\subsetneq V, \{s,t\} \not\subseteq S, S\neq\emptyset;\\
&x(E(S))\leq |S|-2,&\forall S\subsetneq V, \{s,t\} \subseteq S;\\
&x(\delta(\{s\})) = x(\delta(\{t\})) = 1;&\\
&x(\delta(\{v\})) = 2,&\forall v\in V\setminus\{s,t\};\\
&x\geq 0.&
\end{array}
\end{equation*}
\end{obs}

\begin{defn}
For $T\subset V$ and $J\subset E$, $J$ is a \emph{$T$-join} if the set of odd-degree vertices in $G'=(V,J)$ is $T$.
\end{defn}

Edmonds and Johnson \cite{EJ} give a polyhedral characterization of $T$-joins: let $P_T(G)$ be the convex hull of the incidence vectors of the $T$-joins on $G=(V,E)$; $P_T(G)+\mathbb{R}_+^E$ is exactly characterized by
\begin{equation}\label{e:tjoind}
\begin{cases}
y(\delta(S))\geq 1 ,&\forall S\subset V , |S\cap T|\textnormal{ odd};\\
y\in \mathbb{R}_+^E .&
\end{cases}
\end{equation}
We call a feasible solution to \eqref{e:tjoind} a \emph{fractional $T$-join dominator}.

Lastly, the polytope defined by the path-variant Held-Karp relaxation is contained in the spanning tree polytope of the same graph; thus, given a feasible solution $x^*$ to the path-variant Held-Karp relaxation, there exist spanning trees $\mathscr{T}_1,\ldots,\mathscr{T}_k$ and $\lambda_1,\ldots,\lambda_k\in\mathbb{R}_+$ such that $x^*=\sum_{i=1}^k \lambda_i\chi_{\mathscr{T}_i}$ and $\sum_{i=1}^k\lambda_i=1$, where $k$ is bounded by a polynomial. This follows from Gr{\"o}tschel, Lov{\'a}sz, and Schrijver~\cite{GLS}.

\section{Improving upon $5/3$}\label{s:improv}

We present the algorithm for the metric \st path TSP and its analysis in this section.

\subsection{Algorithm}

Given a complete graph $G=(V,E)$ with cost function $c:E\to\mathbb{R}_+$ and the endpoints $s,t\in V$, the algorithm first computes an optimal solution $x^*$ to the path-variant Held-Karp relaxation. Then it decomposes $x^*$ into a convex combination $\sum\lambda_i\chi_{\mathscr{T}_i}$ of polynomially many spanning trees $\mathscr{T}_1,\ldots,\mathscr{T}_k$ with coefficients $\lambda_1,\ldots,\lambda_k\geq 0$; a spanning tree $\mathscr{T}$ is sampled among these spanning trees $\mathscr{T}_i$'s from the probability distribution given by $\lambda_i$'s. This decomposition can be performed in polynomial time, as noted in Section~\ref{s:pre}. Let $T\subset V$ be the set of the vertices with the ``wrong'' parity of degree in $\mathscr{T}$: i.e., $T$ is the set of odd-degree internal points and even-degree endpoints in $\mathscr{T}$. The algorithm finds a minimum $T$-join $J$ and an \st Eulerian path of the multigraph $\mathscr{T}\cup J$. This Eulerian path is shortcut to obtain a Hamiltonian path $H$ between $s$ and $t$; $H$ is the output of the algorithm.

We note that this algorithm can be derandomized by trying each $\mathscr{T}_i$ instead of sampling $\mathscr{T}$. Observe that $\E[c(H)]\leq\rho c(x^*)$ implies that the derandomized algorithm is a deterministic $\rho$-approximation algorithm.

In the rest of this section, we prove the following theorem.

\begin{thm}\label{t:main}
The present algorithm returns a Hamiltonian path between $s$ and $t$ whose expected cost is no more than $\frac{1+\sqrt{5}}{2}c(x^*)$. Therefore, there exists a deterministic $\left(\frac{1+\sqrt{5}}{2}\right)$-approximation algorithm for the \st path TSP.
\end{thm}

\begin{cor}\label{c:ig}
The integrality gap of the path-variant Held-Karp relaxation is at most $\frac{1+\sqrt{5}}{2}$.
\end{cor}

\subsection{Proof of $5/3$-approximation}

In this subsection, we present a simple proof that the present algorithm is a (expected) $5/3$-approximation algorithm. Improved analyses are presented in later subsections based on this simple proof.

We can understand the well-known $2$-approximation algorithm for the circuit TSP and Christofides' $3/2$-approximation algorithm as respectively using the minimum spanning tree and half the Held-Karp solution~\cite{W, SW} as a fractional $T$-join dominator. Let us consider whether $\chi_{\mathscr{T}}$ and $x^*$ can be used to bound the cost of a minimum $T$-join in our case.

It can be seen from \eqref{e:hkpath1} that $\beta x^*$ is a fractional $T$-join dominator for $\beta=1$. If it were not for the \st cuts, the same could be shown for $\beta=\frac{1}{2}$. However, an \st cut may have capacity as low as 1, making it hard to establish the feasibility of $\beta x^*$ for any $\beta<1$.

$\alpha \chi_{\mathscr{T}}$ also is a fractional $T$-join dominator for $\alpha=1$; in this case, however, \st cuts do have some slack. Suppose that an \st cut $(U,\bar U)$ is odd with respect to $T$: i.e., $|U\cap T|$ is odd. Since $U$ contains exactly one of $s$ and $t$, $U$ contains an even number of vertices that have odd degree in $\mathscr{T}$. $|\delta(U)\cap \mathscr{T}|$ is given as the sum of the degrees of the vertices in $U$ minus twice the number of edges within $U$, and is therefore even. This shows $\chi_{\mathscr{T}}(\delta(U))\geq 2$ and hence $\alpha \chi_{\mathscr{T}}$ for $\alpha=\frac{1}{2}$ does not violate \eqref{e:tjoind} as far as \st cuts are concerned. It is the nonseparating cuts that render it difficult to show the feasibility of $\alpha \chi_{\mathscr{T}}$ for $\alpha<1$.

Given the difficulties in these two cases are complementary, it is natural to consider $\alpha \chi_{\mathscr{T}} + \beta x^*$ as a candidate for a fractional $T$-join dominator; Theorem~\ref{t:a53} elaborates this observation.

\begin{thm}\label{t:a53}
$\E[c(H)]\leq\frac{5}{3}c(x^*)$.
\end{thm}
\begin{proof}
Let $y:=\alpha \chi_{\mathscr{T}} + \beta x^*$ for some parameters $\alpha,\beta>0$ to be chosen later. We examine a sufficient condition on $\alpha$ and $\beta$ for $y$ to be a feasible solution to \eqref{e:tjoind}.

It is obvious that $y\geq 0$.

Consider an odd cut $(U,\bar U)$ with respect to $T$: i.e., $|U\cap T|$ is odd. We have $|\delta(U)\cap \mathscr{T}|>0$ from the connectedness of $\mathscr{T}$. Suppose that $(U,\bar U)$ is an $s,t$-cut; then $|\delta(U)\cap \mathscr{T}|$ is even as previously argued. Thus,\begin{eqnarray*}
y(\delta(U)) &=& \alpha |\delta(U)\cap \mathscr{T}|+\beta x^*(\delta(U))\\
&\geq& 2\alpha + \beta
.\end{eqnarray*}

Suppose that $(U,\bar U)$ is nonseparating; then we have $x^*(\delta(U))\geq 2$ from the Held-Karp feasibility, and hence\begin{eqnarray*}
y(\delta(U))&\geq&\alpha|\delta(U)\cap \mathscr{T}|+\beta x^*(\delta(U))\\
&\geq&\alpha+2\beta
.\end{eqnarray*}

Therefore, if $2\alpha+\beta\geq 1$ and $\alpha+2\beta\geq 1$ then $y$ is feasible.
Now we bound the expected cost of $H$:
\begin{eqnarray*}
\E[c(H)] &=& \E[c(\mathscr{T})] + \E[c(J)]\\
&\leq& \E[c(\mathscr{T})] + \E[c(y)]\\
&=& \E[c(\mathscr{T})] + \E[c(\alpha\chi_{\mathscr{T}})] + \E[c(\beta x^*)]\\
&=& (1+\alpha+\beta)c(x^*)
,\end{eqnarray*}where the second line holds since $y$ is a fractional $T$-join dominator. Choose $\alpha=\frac{1}{3}$ and $\beta=\frac{1}{3}$.
\end{proof}

\subsection{First improvement upon $5/3$}

We demonstrate in this subsection that the above analysis can be slightly improved.

Recall that the lower bound on the nonseparating cut capacities of $y$ was given as $\alpha+2\beta$ in the previous analysis; consider perturbing $\alpha$ and $\beta$ by small amount while maintaining $\alpha+2\beta =1$. In particular, if we decrease $\alpha$ by $2\epsilon$ and increase $\beta$ by $\epsilon$, we decrease the expected cost of $y$ by $\epsilon c(x^*)$, without changing $\alpha+2\beta$; that is, if we can fix the possible deficiencies of $y$ in \st cuts with small cost, this perturbation will lead to an improvement in the performance guarantee.

Note that \st cuts with large capacities are not a problem: $(\alpha\chi_{\mathscr{T}}+\beta x^*)(\delta(U))\geq 2\alpha+\beta x^*(\delta(U))$ and thus, if $x^*(\delta(U))$ is large enough, the bound remains greater than one after a small perturbation.

On the other hand, cuts with $x^*(\delta(U))=1$ are also not a concern. $x^*(\delta(U))=\E[|\delta(U)\cap\mathscr{T}|]$, and $|\delta(U)\cap\mathscr{T}|\geq 1$ from the connectedness of $\mathscr{T}$; hence $|\delta(U)\cap\mathscr{T}|$ is identically 1 and $|U\cap T|$ is always even. Formulation~\eqref{e:tjoind} constrains the capacities of only the cuts that are odd with respect to $T$, so the capacity of this particular cut $(U,\bar U)$ will never be constrained. In fact, for an \st cut $(U,\bar U)$, \begin{eqnarray}
\Pr[|U \cap T|\textrm{ is odd}] &\leq& \Pr[|\delta(U)\cap \mathscr{T}|\geq 2]\nonumber\\
&\leq& \E[|\delta(U)\cap \mathscr{T}|] - 1\nonumber\\
&=& x^*(\delta(U))-1\label{e:prob}
.\end{eqnarray}

We will begin with $y\gets \alpha\chi_{\mathscr{T}}+\beta x^*$ for perturbed $\alpha$ and $\beta$, and ensure that $y$ is a fractional $T$-join dominator by adding small fractions of the deficient odd \st cuts. Yet, a cut being odd with small probability as shown by \eqref{e:prob} does not directly connect to its edge being added with small probability, since an edge belongs to many \st cuts. We address this issue by showing that the \st cuts of small capacities are ``almost'' disjoint.

First, consider the \st cuts $(U,\bar U)$ whose capacities are not large enough for $2\alpha+\beta x^*(\delta(U))$ to be readily as large as 1; the following definition captures this idea. Let $\tau:=\frac{1-2\alpha}{\beta}-1$.

\begin{defn}
For some $0< \tau\leq 1$, an \st cut $(U,\bar U)$ is called \emph{$\tau$-narrow} if $x^*(\delta(U))<1+\tau$.
\end{defn}

The following lemma shows that $\tau$-narrow cuts do not cross.

\begin{lemma}\label{l:noncross}
Let $0< \tau\leq 1$. For $U_1 \ni s$ and $U_2 \ni s$, if both $(U_1 ,\bar U_1)$ and $(U_2 ,\bar U_2)$ are $\tau$-narrow, then $U_1\subset U_2$ or $U_2\subset U_1$.
\end{lemma}
\begin{proof}
Suppose not. Then both $U_1\setminus U_2$ and $U_2\setminus U_1$ are nonempty and\begin{equation}\label{e:l:noncross:1}
x^*(\delta(U_1))+x^*(\delta(U_2))\geq x^*(\delta(U_1\setminus U_2))+x^*(\delta(U_2\setminus U_1))\geq 2+2=4
;\end{equation}on the other hand,\begin{equation*}
x^*(\delta(U_1))+x^*(\delta(U_2)) < 2+2\tau \leq 4
,\end{equation*}contradicting \eqref{e:l:noncross:1}.
\end{proof}

Lemma \ref{l:noncross} shows that the $\tau$-narrow cuts constitute a layered structure, as illustrated in Figure~\ref{f:layer}:
\begin{figure}
\center
\includegraphics[width=350pt]{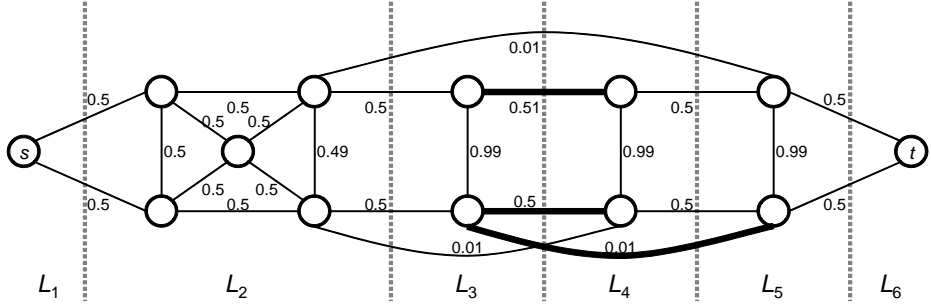}
\caption{$0.05$-narrow cuts of a feasible Held-Karp solution ($\ell=6$). $F_3$ is marked with thick edges.}
\label{f:layer}
\end{figure}

\begin{cor}\label{c:layer}
There exists a partition $L_1,\ldots,L_\ell$ of $V$ such that\begin{enumerate}
\item $L_1=\{s\}$, $L_\ell=\{t\}$, and
\item $\{U|(U,\bar U)\textrm{ is }\tau\textrm{-narrow, }s\in U\}=\{U_i |1\leq i<\ell\}$, where $U_i := \cup_{k=1}^i L_k$.
\end{enumerate}
\end{cor}

Let $L_{\leq i}$ denote $\cup_{k=1}^i L_k$ and $L_{\geq i}$ denote $\cup_{k=i}^\ell L_k$. $U_i = L_{\leq i}$.

Now we show that $\tau$-narrow cuts are almost disjoint: for each $\tau$-narrow cut $(U_i,\bar U_i)$, we can choose $F_i\subset\delta(U_i)$ that occupies a large portion of $\delta(U_i)$ and mutually disjoint.

\begin{defn}
$F_i := E(L_i, L_{\geq i+1})$.
\end{defn}

\begin{lemma}\label{l:largedisjoint}
For each $\tau$-narrow cut $(U_i,\bar U_i)$, $x^*(F_i)>\frac{1-\tau+x^*(\delta(U_i))}{2} \geq 1-\frac{\tau}{2}$.
\end{lemma}
\begin{proof}
The lemma holds trivially for $i=1$. Suppose $2\leq i\leq \ell-1$. We have\begin{equation}
1+\tau > x^*(\delta(U_{i-1})) = x^*(E(L_{\leq i-1}, L_i)) + x^*(E(L_{\leq i-1}, L_{\geq i+1}))\label{e:1}
\end{equation}and\begin{equation}
x^*(\delta(U_i)) = x^*(E(L_i, L_{\geq i+1})) + x^*(E(L_{\leq i-1}, L_{\geq i+1}))\label{e:2}
.\end{equation}
From \eqref{e:1} and \eqref{e:2},\[
x^*(\delta(U_i))-1-\tau < x^*(E(L_i, L_{\geq i+1})) - x^*(E(L_{\leq i-1}, L_i))
;\]on the other hand,\[
2 \leq x^*(\delta(L_i)) = x^*(E(L_i, L_{\geq i+1})) + x^*(E(L_{\leq i-1}, L_i))
.\]Thus,\[
x^*(F_i) = x^*(E(L_i, L_{\geq i+1})) > \frac{1-\tau+x^*(\delta(U_i))}{2} \geq 1-\frac{\tau}{2}
.\]
\end{proof}

It is obvious that $F_i$'s are disjoint and $F_i\subset\delta(U_i)$. For each $\tau$-narrow cut $U_i$, we define $f^*_{U_i}$ as\[
(f^*_{U_i})_e :=\begin{cases}
x_e^*&\textrm{if }e\in F_i,\\
0&\textrm{otherwise}
.\end{cases}
\]

\begin{thm}\label{t:qi}
$\E[c(H)]\leq 1.6577 c(x^*)$.
\end{thm}
\begin{proof}
Let\[
y:=\alpha\chi_{\mathscr{T}}+\beta x^*+\sum_{i:|U_i \cap T|\mathrm{~is~odd,~}1\leq i<\ell} \frac{1-(2\alpha+\beta)}{1-\frac{\tau}{2}}f^*_{U_i}
,\]for $\alpha=0.30$, $\beta=0.35$ and $\tau=\frac{1-2\alpha}{\beta}-1=\frac{1}{7}$. We claim $y$ is a fractional $T$-join dominator. It is obvious that $y\geq 0$, and we have argued that $y(\delta(U))\geq 1$ for nonseparating $(U,\bar U)$. Suppose $(U,\bar U)$ is an \st cut with $|U\cap T|$ odd. If $(U,\bar U)$ is not $\tau$-narrow, then\begin{eqnarray*}
y(\delta(U)) &\geq& \alpha |\delta(U)\cap \mathscr{T}|+\beta x^*(\delta(U))\\
&\geq& 2\alpha + \beta(1+\tau)\\
&=&1
.\end{eqnarray*}If $(U,\bar U)$ is $\tau$-narrow, then\begin{eqnarray*}
y(\delta(U)) &\geq& \alpha |\delta(U)\cap \mathscr{T}|+\beta x^*(\delta(U)) + \frac{1-(2\alpha+\beta)}{1-\frac{\tau}{2}} f^*_U(\delta(U))\\
&\geq& 2\alpha + \beta + \frac{1-(2\alpha+\beta)}{1-\frac{\tau}{2}} \left(1-\frac{\tau}{2}\right)\\
&=&1
.\end{eqnarray*}

Thus $y$ is a fractional $T$-join dominator. Now it remains to bound the expected cost of $H$. Let $\displaystyle A:=\frac{1-(2\alpha+\beta)}{1-\frac{\tau}{2}}$. \begin{eqnarray*}
\E[c(H)] &=& \E[c(\mathscr{T})] + \E[c(J)]\\
&\leq& \E[c(\mathscr{T})] + \E[c(y)]\\
&=& \E[c(\mathscr{T})] + \E[c(\alpha\chi_{\mathscr{T}})] + \E[c(\beta x^*)] +
\E\left[c\left(\sum_{i:|U_i \cap T|\mathrm{~is~odd,~}1\leq i<\ell} A\cdot f^*_{U_i}\right)\right]\\
&=& (1+\alpha+\beta)c(x^*)+ c\left(\sum_{i=1}^{\ell-1} \Pr[|U_i \cap T|\textrm{ is odd}] \cdot A\cdot f^*_{U_i}\right)
.\end{eqnarray*}From \eqref{e:prob},\begin{eqnarray*}
\E[c(H)] &\leq& (1+\alpha+\beta)c(x^*)+ \tau A c\left(\sum_{i=1}^{\ell-1} f^*_{U_i}\right)\\
&\leq& \left(1+\alpha+\beta+ \tau A \right )c(x^*)
,\end{eqnarray*}where the last line follows from the disjointness of $F_i$. Note that $1+\alpha+\beta+ \tau A <1.6577$.
\end{proof}

\subsection{A tighter analysis}

In the previous analysis, we separately bounded the probability that a $\tau$-narrow cut is odd, the deficit of the cut, and $f^*_U(\delta(U))$; moreover, we used $1-\frac{\tau}{2}$ instead of $\frac{1-\tau+x^*(\delta(U_i))}{2}$ from Lemma~\ref{l:largedisjoint}. These observations lead to some improvement, as shown in the following theorem.

\begin{thm}\label{t:iint}
$\E[c(H)]\leq\frac{9-\sqrt{33}}{2} c(x^*)$.
\end{thm}
\begin{proof}
Let $\displaystyle b_i:=\frac{1-\tau+x^*(\delta(U_i))}{2}$ denote the lower bound of $f^*_{U_i}(\delta(U_i))$ given by Lemma~\ref{l:largedisjoint}.

Let\[
y:= \alpha\chi_{\mathscr{T}}+\beta x^*+\sum_{i:|U_i \cap T|\mathrm{~is~odd,~}1\leq i<\ell} \frac{1-\{2\alpha+\beta x^*(\delta(U_i))\}}{b_i}f^*_{U_i}
,\]where $\alpha$ and $\beta$ are to be chosen later; $\tau:=\frac{1-2\alpha}{\beta}-1$. As in the previous subsection, $\{U_i\}$ and $\{L_i\}$ denote the $\tau$-narrow cuts and their layered structure. Assume $\frac{1}{3}\leq \beta\leq\frac{1}{2}$ and $1-2\beta\leq\alpha\leq\frac{1-\beta}{2}$.

A similar argument as in Theorem~\ref{t:qi} proves that $y$ is a fractional $T$-join dominator; it can also be shown that\begin{eqnarray}
\E[c(H)] &\leq& (1+\alpha+\beta)c(x^*)+ c\left(\sum_{i=1}^{\ell-1} \Pr[|U_i \cap T|\textrm{ is odd}] \frac{1-\{2\alpha+\beta x^*(\delta(U_i))\}}{b_i}f^*_{U_i}\right) \nonumber\\
&\leq& (1+\alpha+\beta)c(x^*)+ c\left(\sum_{i=1}^{\ell-1} \{x^*(\delta(U_i))-1\} \frac{1-\{2\alpha+\beta x^*(\delta(U_i))\}}{b_i}f^*_{U_i}\right)\nonumber\\
&\leq& (1+\alpha+\beta)c(x^*)+ \left[\max_{0\leq\omega\leq \tau}\left(\omega\frac{1-\{2\alpha+\beta(1+\omega)\}}{1-\frac{\tau}{2}+\frac{\omega}{2}}\right)\right] c\left(\sum_{i=1}^{\ell-1} f^*_{U_i}\right)\nonumber\\
&\leq& \left\{ 1+\alpha+\beta+ \max_{0\leq\omega\leq \tau}\left(\omega\frac{1-\{2\alpha+\beta(1+\omega)\}}{1-\frac{\tau}{2}+\frac{\omega}{2}}\right)\right\}c(x^*)\label{e:iint:1}
.\end{eqnarray}

Let $\displaystyle R(\omega):=\omega\frac{1-\{2\alpha+\beta(1+\omega)\}}{1-\frac{\tau}{2}+\frac{\omega}{2}} = \frac{\omega[1-\{2\alpha+\beta(1+\omega)\}]}{\frac{3}{2}-\frac{1}{2\beta}+\frac{\alpha}{\beta}+\frac{\omega}{2}}$. We have\[
R'(\omega )=\frac{-\frac{\beta}{2}\omega^2+(1-2\alpha-3\beta)\omega+\left(2-4\alpha-\frac{3}{2}\beta-\frac{1}{2\beta}+\frac{2\alpha}{\beta}-\frac{2\alpha^2}{\beta}\right)}{\left(\frac{3}{2}-\frac{1}{2\beta}+\frac{\alpha}{\beta}+\frac{\omega}{2}\right)^2}
\]and the unique solution to\[
R'(\omega)=0\quad(0\leq\omega\leq\frac{1-2\alpha}{\beta}-1)
\]is\[
\omega=\omega_0 :=\frac{1}{\beta}\left(1-2\alpha-3\beta+\sqrt{(-2\beta)(1-2\alpha-3\beta)}\right)
.\]Since $R(\omega)\geq 0$ for $0\leq\omega\leq\frac{1-2\alpha}{\beta}-1$ and $R(0)=R(\frac{1-2\alpha}{\beta}-1)=0$, $R(\omega)$ is maximized at $\omega=\omega_0$; hence, from \eqref{e:iint:1},\[
E[c(H)]\leq\left(5\alpha+11\beta-1-4\sqrt{(-2\beta)(1-2\alpha-3\beta)}\right)c(x^*)
.\]
Choose $\alpha=\frac{1}{\sqrt{33}}$, $\beta=\frac{1}{2}-\frac{1}{2\sqrt{33}}$ and we obtain\[
E[c(H)]\leq \frac{9-\sqrt{33}}{2}c(x^*)
.\]
\end{proof}

\subsection{Proof of $\left(\frac{1+\sqrt{5}}{2}\right)$-approximation}

In this final subsection, we show that $\E[c(H)]\leq\frac{1+\sqrt{5}}{2}c(x^*)$, proving Theorem~\ref{t:main} and Corollary~\ref{c:ig}.

In the previous analyses, $F_i$'s serve as ``representatives'' of $\tau$-narrow cuts. These representatives are useful since they have large weights while being disjoint. We improve the performance guarantee by introducing a new set of representatives that are ``fractionally disjoint''. Note that the three key properties of $\{f^*_{U_i}\}$ used in the proof of Theorem~\ref{t:iint} are:\begin{enumerate}
\item $f^*_{U_i}\geq 0$ for all $i$;
\item $\sum_{i=1}^{\ell-1} f^*_{U_i} \leq x^*$; and
\item $f^*_{U_i}(\delta(U_i))\geq \frac{1-\tau+x^*(\delta(U_i))}{2}$ for all $i$.
\end{enumerate}
$\{f^*_{U_i}\}$ chosen in the previous analyses also satisfies that, for any given $e\in E$, $\left(f^*_{U_i}\right)_e \neq 0$ for at most one $i$. However, this was not a useful property in the analysis; Lemma~\ref{l:fd} states that, by relaxing the definition of disjointness, we can choose $\{\hat f ^*_{U_i}\}$ that have larger weights. The definitions of $\tau$, $\{U_i\}$ and $\{L_i\}$ are unchanged.

\begin{lemma}
\label{l:fd}
There exists a set of vectors $\{\hat f _{U_i}^*\}_{i=1}^{\ell-1}$ satisfying:\begin{enumerate}
\item $\hat f _{U_i}^*\in\mathbb{R}_+^E$ for all $i$;
\item $\sum_{i=1}^{\ell-1} \hat f _{U_i}^* \leq x^*$; and
\item $\hat f _{U_i}^*(\delta(U_i))\geq 1$ for all $i$.
\end{enumerate}
\end{lemma}

This lemma is proven later; based on it, Lemma~\ref{l:fd} proves the desired performance guarantee.

\begin{lemma}\label{l:fin}
$\E[c(H)]\leq\frac{1+\sqrt{5}}{2}c(x^*)$.
\end{lemma}
\begin{proof}
Let\[
y:=\alpha\chi_{\mathscr{T}}+\beta x^*+\sum_{i:|U_i \cap T|\mathrm{~is~odd,~}1\leq i<\ell} \left[1-\{2\alpha+\beta x^*(\delta(U_i))\}\right]\hat f ^*_{U_i}
,\]where $\alpha$ and $\beta$ are parameters to be chosen later, satisfying\begin{equation}\label{e:later1}
\frac{1}{3}\leq \beta\leq\frac{1}{2} \quad\textrm{and}\quad 1-2\beta\leq\alpha\leq\frac{1-\beta}{2}
.\end{equation}

By following the same argument as in Theorem~\ref{t:iint}, we can easily show that $y$ is a fractional $T$-join dominator; the only slight difference is when $(U,\bar U)$ is $\tau$-narrow and $|U\cap T|$ is odd, where we have\begin{eqnarray*}
y(\delta(U)) &\geq& \alpha |\delta(U)\cap \mathscr{T}|+\beta x^*(\delta(U)) + \left[1-\{2\alpha+\beta x^*(\delta(U_i))\}\right] \hat f ^*_U(\delta(U))\\
&\geq& 2\alpha + \beta x^*(\delta(U)) + \left[1-\{2\alpha+\beta x^*(\delta(U_i))\}\right] \cdot 1\\
&=&1
,\end{eqnarray*}from the first and the third properties of Lemma~\ref{l:fd}. Hence, $y$ is a fractional $T$-join dominator.

Now it remains to bound $\E[c(H)]$.\begin{eqnarray}
\E[c(H)] &\leq& \E[c(\mathscr{T})] + \E[c(y)]\nonumber\\
&=& (1+\alpha+\beta)c(x^*)+ c\left(\sum_{i=1}^{\ell-1} \Pr[|U_i \cap T|\textrm{ is odd}] \left[1-\left\{2\alpha+\beta x^*(\delta(U_i))\right\}\right] \hat f ^*_{U_i}\right)\nonumber\\
&\leq& (1+\alpha+\beta)c(x^*)+ c\left(\sum_{i=1}^{\ell-1} \{x^*(\delta(U_i))-1\} \left[1-\left\{2\alpha+\beta x^*(\delta(U_i))\right\}\right] \hat f ^*_{U_i}\right)\nonumber\\
&\leq& (1+\alpha+\beta)c(x^*)+ \left\{\max_{0\leq\omega\leq \tau} \omega \left[1-\left\{2\alpha+\beta (1+\omega))\right\}\right] \right\} c\left(\sum_{i=1}^{\ell-1} \hat f ^*_{U_i}\right).\label{e:later2}
\end{eqnarray}From the second property of Lemma~\ref{l:fd},\begin{eqnarray*}
\E[c(H)] &\leq& \left\{ 1+\alpha+\beta+ \max_{0\leq\omega\leq \tau} \omega \left[1-\left\{2\alpha+\beta (1+\omega))\right\}\right] \right\}c(x^*)\\
&=& \left\{ 1+\alpha+\beta+ \max_{0\leq\omega\leq \tau} \beta\omega(\tau-\omega) \right\}c(x^*)\\
&=& \left\{ 1+\alpha+\beta+ \frac{(1-2\alpha-\beta)^2}{4\beta} \right\}c(x^*)
.\end{eqnarray*}
We choose $\alpha=1-\frac{2}{\sqrt{5}}$ and $\beta=\frac{1}{\sqrt{5}}$.
\end{proof}

\begin{proof}[Proof of Lemma~\ref{l:fd}]
Consider an auxiliary flow network illustrated in Figure~\ref{f:fd}, consisting of the source $v^{\mathsf{source}}$, sink $v^{\mathsf{sink}}$, a node $v^{\mathsf{cut}}_U$ for each $\tau$-narrow cut $U$, and a node $v^{\mathsf{edge}}_e$ for each edge $e$ in one or more $\tau$-narrow cuts. The network has arcs of:\begin{enumerate}
\item capacity 1 from $v^{\mathsf{source}}$ to $v^{\mathsf{cut}}_U$ for every $\tau$-narrow cut $U$;
\item capacity $\infty$ from $v^{\mathsf{cut}}_U$ to $v^{\mathsf{edge}}_e$ for every $e\in\delta(U)$, for all $U$;
\item capacity $x^*_e$ from $v^{\mathsf{edge}}_e$ to $v^{\mathsf{sink}}$ for every $v^{\mathsf{edge}}_e$.
\end{enumerate} Let $g$ be this capacity function.

\begin{figure}
\center
\includegraphics[width=320pt]{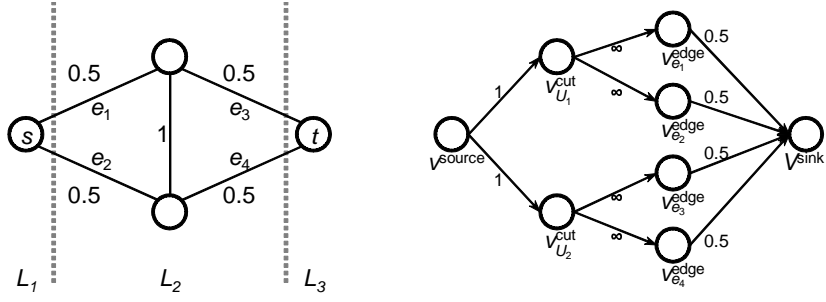}
\caption{A feasible Held-Karp solution ($\ell=3$) and its corresponding flow network.}
\label{f:fd}
\end{figure}

Let $(S,\bar S)$ be an arbitrary cut on this flow network, where $v^{\mathsf{source}}\in S$. We claim the cut capacity of $(S,\bar S)$ is at least $\ell-1$.

Suppose there exists a $\tau$-narrow cut $U$ and $e\in\delta(U)$ such that $v^{\mathsf{cut}}_U\in S$ and $v^{\mathsf{edge}}_e\notin S$; the cut capacity is then $\infty$. So assume from now that (abusing the notation) every edge in any $\tau$-narrow cut in $S$ is also in $S$. Let $S\cap\{v^{\mathsf{cut}}_{U_i}|1\leq i<\ell\}=\{v^{\mathsf{cut}}_{U_{i_1}}, v^{\mathsf{cut}}_{U_{i_2}}, \ldots, v^{\mathsf{cut}}_{U_{i_k}}\}$ for some $1\leq i_1 < i_2 < \ldots < i_k < \ell$. The cut capacity is then at least\begin{eqnarray*}
&&\sum_{v^{\mathsf{cut}}_U\notin S}g(v^{\mathsf{source}},v^{\mathsf{cut}}_U)+
\sum_{e:\exists v^{\mathsf{cut}}_U\in S\ e\in\delta(U)}g(v^{\mathsf{edge}}_e,v^{\mathsf{sink}})\\
&=& (\ell-1-k) + \sum_{e:\exists v^{\mathsf{cut}}_U\in S\ e\in\delta(U)} x_e^*
;\end{eqnarray*}if $k=0$, the claim holds; the claim also holds for $k=1$ since $x^*(\delta(U_{i_1}))\geq 1$. Suppose $k\geq 2$ (see Figure~\ref{f:fd2}).\begin{eqnarray*}
\sum_{e:\exists v^{\mathsf{cut}}_U\in S\ e\in\delta(U)} x_e^*
&=&\frac{1}{2}\left[x^*(\delta(U_{i_1}))+\sum_{j=2}^k x^*(\delta(U_{i_j}\setminus U_{i_{j-1}}))+x^*(\delta(V\setminus U_{i_k}))\right]\\
&\geq& \frac{1}{2}\left[1+2(k-1)+1\right]\\
&=& k
,\end{eqnarray*}proving the claim.

\begin{figure}
\center
\includegraphics[width=280pt]{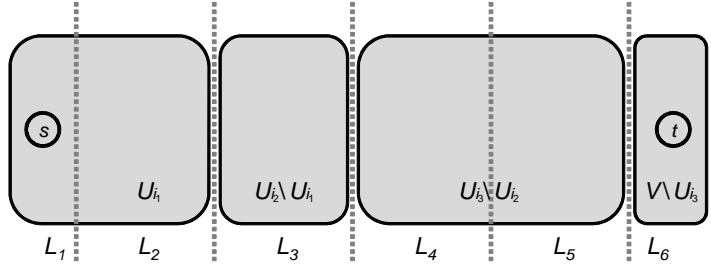}
\caption{Schematic diagram: $\ell=6$, $k=3$, $i_1=2$, $i_2=3$, and $i_3=5$.}
\label{f:fd2}
\end{figure}

Thus the maximum flow on this flow network is of value at least $\ell-1$. Consider a maximum flow; this flow saturates all the edges from $v^{\mathsf{source}}$ to $v^{\mathsf{cut}}_U$, since the cut capacity of $(\{v^{\mathsf{source}}\},\overline{\{v^{\mathsf{source}}\}})$ is $\ell-1$. Now, for each $\tau$-narrow cut $U$, define $(\hat f ^*_U)_e$ as the flow from $v^{\mathsf{cut}}_U$ to $v^{\mathsf{edge}}_e$ if $e\in \delta(U)$, and 0 otherwise. Then the first property is satisfied from the definition of flow; the second property is satisfied from the capacity constraints on $v_e^{\mathsf{edge}}$ to $v^{\mathsf{sink}}$; lastly, the third property is satisfied since every edge from $v^{\mathsf{source}}$ to $v^{\mathsf{cut}}_U$ is saturated.
\end{proof}

\section{Application to other problems}\label{s:appl}

In this section, we exhibit how the present results can be applied to other problems to obtain approximation algorithms with better performance guarantees than the current best known and improved LP integrality gap upper bounds.

\subsection{Prize-collecting \st path problem}\label{ss:prize}

We discuss the prize-collecting \st path problem in this subsection.

\begin{defn}[Metric prize-collecting \st path problem]
Given a complete graph $G=(V,E)$ with $s,t\in V$, metric edge cost function $c:E\to\mathbb{R}_+$, and vertex prize $\pi:V\to\mathbb{R}_+$, the metric prize-collecting \st path problem is to find a simple \st path $P$ that minimizes the sum of the path cost and the total prize ``missed'', i.e., $c(P)+\pi(V\setminus V(P))$.
\end{defn}

The \st path TSP can be considered as a special case of the prize-collecting \st path problem, where $\pi(v)=\infty$ for all $v\in V$.

Archer et al.~\cite{ABHK} use the path-variant Christofides' algorithm~\cite{H} as a black box to obtain a $\frac{241}{121}$-approximation algorithm for the metric prize-collecting \st path problem. $\frac{241}{121}<1.9918$.

\begin{thm}[Archer et al.~\cite{ABHK}]\label{t:abhk}
Given a $\rho$-approximation algorithm $\mathscr{A}$ for the  metric \st path TSP, one can obtain a $\left(2-\left(\frac{2-\rho}{2+\rho}\right)^2\right)$-approximation algorithm for the metric prize-collecting \st path problem that uses $\mathscr{A}$ as a black box.
\end{thm}

This theorem, combined with Theorem~\ref{t:main}, readily produces an improvement. $\frac{1+4\sqrt{5}}{5}<1.9889$.

\begin{cor}
There exists a $\left(\frac{1+4\sqrt{5}}{5}\right)$-approximation algorithm for the metric prize-collecting \st path problem.
\end{cor}

However, as the performance guarantee established by Theorem~\ref{t:main} is in terms of the Held-Karp optimum, the theorem enables a further improvement via an analysis analogous to Goemans~\cite{G:pc}. For the metric prize-collecting traveling salesman problem, Goemans~\cite{G:pc} combines the LP rounding algorithm due to Bienstock et al.~\cite{BGSW} and the primal-dual algorithm of Goemans \& Williamson~\cite{GW} (with the observation of \cite{CRW} and \cite{ABHK}) to achieve the best performance guarantee known for the problem.

One obstacle in applying this approach to the prize-collecting \st path problem is that, unlike the circuit-variant Held-Karp relaxation, the path-variant Held-Karp relaxation cannot be written as a set of edge-connectivity requirements between the pairs of vertices: the relaxation requires nonseparating cuts to have capacity of at least 2, whereas the edge connectivity between any two vertices can be as low as 1 in both a feasible Held-Karp solution and a (integral) Hamiltonian path. We will show that, despite this fact, the parsimonious property~\cite{GB} still can be used, and will analyze the performance guarantee given by the approach.

We start with the following LP relaxation of the problem:\begin{equation}\label{e:pca}
\begin{array}{lll}
\textrm{minimize}&c(x)+\pi(\mathbf{1}-y)&\\
\textrm{subject to}&x(\delta(S))\geq 1,&\forall S\subsetneq V, |S\cap\{s,t\}|=1;\\
&x(\delta(S))\geq 2y_v,&\forall S\subsetneq V, S\cap\{s,t\}=\emptyset\quad \forall v\in S;\\
&x(\delta(\{s\})) = x(\delta(\{t\})) = 1;&\\
&x(\delta(\{v\})) = 2y_v,&\forall v\in V\setminus\{s,t\};\\
&x_e\geq 0,&\forall e\in E;\\
&0\leq y_v\leq 1,&\forall v\in V\setminus\{s,t\};
\end{array}
\end{equation}where $\mathbf{1}$ denotes the all-1 vector in $V\in\mathbb{R}_+^{V\setminus\{s,t\}}$. It can be easily verified that this is a relaxation of the prize-collecting \st path problem.

Given $V'\subset V\setminus\{s,t\}$, consider a related problem of finding a minimum \st path on $G$ that visits all the vertices in $V'$, and only those vertices. The following LP is a relaxation to this problem:\begin{equation}\label{e:pcb}
\begin{array}{lll}
\textrm{minimize}&c(x)&\\
\textrm{subject to}&x(\delta(S))\geq 1,&\forall S\subsetneq V, |S\cap\{s,t\}|=1;\\
&x(\delta(S))\geq 2,&\forall S\subsetneq V, S\cap\{s,t\}=\emptyset, S\cap V'\neq\emptyset;\\
&x(\delta(\{s\})) = x(\delta(\{t\})) = 1;&\\
&x(\delta(\{v\})) = 2,&\forall v\in V';\\
&x(\delta(\{v\})) = 0,&\forall v\in V\setminus\{s,t\}\setminus V';\\
&x_e\geq 0,&\forall e\in E.
\end{array}
\end{equation}

\begin{obs}\label{o:pchk}
Let $G'=(V'\cup\{s,t\},E')$ be the subgraph of $G$ induced by $V'\cup\{s,t\}$. Projecting a feasible solution to \eqref{e:pcb} to $E'$ yields a feasible solution to the path-variant Held-Karp relaxation for $G'$.
\end{obs}

The following lemma shows that we can use the parsimonious property.

\begin{lemma}\label{l:pcpp}
The optimal solution value to \eqref{e:pcb} is equal to the optimal solution value to the following relaxation without the degree constraints:\begin{equation}\label{e:pcc}
\begin{array}{lll}
\textrm{minimize}&c(x)&\\
\textrm{subject to}&x(\delta(S))\geq 1,&\forall S\subsetneq V, |S\cap\{s,t\}|=1;\\
&x(\delta(S))\geq 2,&\forall S\subsetneq V, S\cap\{s,t\}=\emptyset, S\cap V'\neq\emptyset;\\
&x_e\geq 0,&\forall e\in E.
\end{array}
\end{equation}
\end{lemma}
\begin{proof}
Let $G=(V,E)$. It suffices to show that, given a feasible solution $x^*$ to \eqref{e:pcc}, how to construct a feasible solution to \eqref{e:pcb} whose cost is no greater than $c(x^*)$.

We will extend the graph (and $x^*$) so that the relaxation (almost) becomes a set of edge-connectivity requirements between pairs of vertices, and then use a similar approach as in Bienstock et al.~\cite{BGSW}, along with the following lemma:
\begin{lemma}[\cite{BGSW}]\label{l:pcsplit}
Let $G=(V,E)$ be an Eulerian multigraph. Suppose that, for some $U\subset V$ and $v\in V$, any two vertices in $U$ other than $v$ are $k$-edge-connected. Let $x$ be an arbitrary neighbor of $v$; then, there exists a neighbor $y$ of $v$ such that\begin{enumerate}
\item $x\neq y$; and
\item any two vertices in $U$ other than $v$ are still $k$-edge-connected after splitting $(x,v)$ and $(y,v)$: i.e., replacing $(x,v)$ and $(y,v)$ (one copy each) with $(x,y)$.
\end{enumerate}
\end{lemma}

Without loss of generality, we can assume $x^*$ is rational.

Now we add three new vertices to the graph: $s'$, $t'$ and $u$. We set $c(s',v)=c(s,v)$ and $c(t',v)=c(t,v)$ for all $v$; $c(s',s)=c(t',t)=0$: $s'$ and $t'$ will be the ``proxy'' of $s$ and $t$. We do not define the cost between $u$ and other vertices: these costs do not affect the rest of the analysis. However, for notational convenience, we set these costs to be zero, potentially violating the triangle inequality. Let $\bar G=(\bar V,\bar E)$ be this extended graph.

We extend $x^*$ into $\bar x^*$ as well: $\bar x^*(s,s')=\bar x^*(s',u)=\bar x^*(u,t')=\bar x^*(t',t)=1$, and all other newly added edges are set to zero. Note that the (fractional) degree of $s'$, $t'$ and $u$ are 2.

Let $\bar V':=V'\cup\{s',t',u\}$; we claim that any two vertices in $\bar V'$ are 2-edge-connected.
\begin{claim}\label{c:pc2}
For any $S\subset \bar V$ such that $\bar V'\cap S\neq \emptyset$ and $\bar V'\setminus S\neq \emptyset$, $\bar x^*(\delta(S))\geq 2$.
\end{claim}
\begin{proof}
Without loss of generality, assume $s\in S$. If $t\notin S$, then at least one edge of the path $P:s-s'-u-t'-t$ is in $\delta(S)$; thus,\[
\bar x^*(\delta(S))\geq x^*(\delta_G(S\cap V))+\bar x^*(\delta(S)\cap P)\geq 1+1
.\]

Suppose $t\in S$. If $\{s',u,t'\}\setminus S\neq\emptyset$ then $|\delta(S)\cap P|\geq 2$; hence,\[
\bar x^*(\delta(S))\geq \bar x^*(\delta(S)\cap P)\geq 2
.\]Otherwise, $V'\setminus S=\bar V'\setminus S\neq\emptyset$ and thus,\[
\bar x^*(\delta(S))\geq x^*(\delta_G(S\cap V))\geq 2
,\]since $(S\cap V)\cap V'\subsetneq V'$.
\end{proof}

Now scale $\bar x^*$ by some large constant $C$ so that $\bar z^*:=C\bar x^*$ is integral and, in the multigraph on $\bar V$ whose edge multiplicities are given by $\bar z^*$, the degree of every vertex is even. Note that any two vertices in $\bar V'$ are $2C$-edge-connected in this multigraph.

Let $\phi:=\sum_{v\in\bar V'} [\bar z^*(\delta(v))-2C] + \sum_{v\in \bar V\setminus \bar V'} \bar z^*(\delta(v))$; $\phi$ is an even integer. We will modify $\bar z^*$ until $\phi$ reaches 0: in particular, we split two edges in the multigraph so that\begin{enumerate}[(i)]
\item $\phi$ decreases by 2; \label{i:pc1}
\item $c(\bar z^*)$ do not increase; \label{i:pc2}
\item any two vertices in $\bar V'$ are $2C$-edge-connected; \label{i:pc3}
\item the degrees of $s'$, $t'$ and $u$ all remain $2C$; \label{i:pc4}
\item the only edges incident to $u$ are $(s',u)$ and $(u,t')$; and \label{i:pc5}
\item every vertex has even degree and hence the connected component containing $\bar V'$ is Eulerian. \label{i:pc6}
\end{enumerate}It is clear that the invariants \eqref{i:pc3} through \eqref{i:pc6} initially hold.

If there exists an edge that is not reachable from any vertex in $\bar V'$, we can remove all such edges without violating any of the conditions ($\phi$ may decrease by more than 2).

If there exists $v\in\bar V\setminus\bar V'$ such that $\bar z^*(\delta(v))> 0$, then we apply Lemma~\ref{l:pcsplit} to pick two incident edges to split. Note that $v\notin\{s',t',u\}$ since $s',t',u\in \bar V'$. \eqref{i:pc3} is maintained from the lemma. Splitting does not change the degree of any vertex other than $v$; hence \eqref{i:pc1}, \eqref{i:pc4} and \eqref{i:pc6} are satisfied. Neither of the chosen edges is incident to $u$, as can be seen from \eqref{i:pc5}; thus, \eqref{i:pc5} is maintained and \eqref{i:pc2} follows from the triangle inequality.

Otherwise, we choose $v\in \bar V'$ such that $\bar z^*(\delta(v))> 2C$. $\bar z^*(\delta(v))\geq 2C+2$ from \eqref{i:pc6}. Again $v\notin\{s',t',u\}$ from \eqref{i:pc4}; we can similarly verify all properties in this case as well.

Once $\phi$ reaches $0$, we remove $u$ and its incident edges. None of these edges got split during the process: this is the reason why the cost of these edges can be left undefined.

Note that the degree of $s$ and $t$ now are $0$, whereas $s'$ and $t'$ are $1$. Concatenate $s$ and $s'$, and $t$ and $t'$, respectively; we scale this multigraph back by $1/C$ to obtain a feasible solution to \eqref{e:pcb} whose cost is no greater than $c(x^*)$.
\end{proof}

We are now ready to apply the analyses of Goemans~\cite{G:pc} and Bienstock et al.~\cite{BGSW}. Let $x^*$ and $y^*$ be an optimal solution to \eqref{e:pca}.
\begin{lemma}\label{l:prounding}
Let $\mathscr{A}^\rho$ be an approximation algorithm for the \st path TSP that produces a path of cost at most $\rho$ times the Held-Karp optimum. Let $V_\gamma=\{v|y_v^*\geq \gamma\}$ for some $0<\gamma\leq 1$. Running $\mathscr{A}^\rho$ on the subgraph $G_\gamma$ induced by $V_\gamma\cup\{s,t\}$ yields a path $P$ with $c(P)\leq\frac{\rho}{\gamma}c(x^*)$.
\end{lemma}
\begin{proof}The proof is basically the same as \cite{BGSW}. Observe that $\frac{x^*}{\gamma}$ is a feasible solution to \eqref{e:pcc}, as can be seen from \eqref{e:pca} and \eqref{e:pcc}. From Lemma~\ref{l:pcpp} and Observation~\ref{o:pchk}, the Held-Karp optimum for $G_\gamma$ is of cost no greater than $c(\frac{x^*}{\gamma})$.
\end{proof}

The primal-dual algorithm of Chaudhuri et al.~\cite{CGRT} can be used to obtain the following performance guarantee for the metric prize-collecting \st path problem.
\begin{lemma}[\cite{CGRT, ABHK}]\label{l:cgrtpg}
There exists a polynomial-time algorithm $\mathscr{A}_\mathsf{PD}$ that produces an \st path $P$ satisfying\[
c(P)+\pi(V\setminus V(P))\leq 2c(x^*)+\pi(\mathbf{1}-y^*)
.\]
\end{lemma}

Now, the combined algorithm is as follows: let $a := e^{1-\frac{2}{\rho}}$ and $p := \frac{1+\rho\ln a}{2-a+\rho\ln a}$. The algorithm runs $\mathscr{A}_\mathsf{PD}$ with probability $p$; otherwise, it computes an optimal solution $x^*$ and $y^*$ to \eqref{e:pca}, samples $\gamma$ uniformly at random from $(a,1)$, and run $\mathscr{A}^\rho$ on the subgraph induced by $V_\gamma=\{v|y_v^*\geq \gamma\}$.

This algorithm can be derandomized since there are only $O(|V|)$ different $V_\gamma$'s possible.

\begin{thm}\label{t:pcfinal}
Let $\mathscr{A}^\rho$ be an approximation algorithm for the \st path TSP that produces a path of cost at most $\rho$ times the Held-Karp optimum, for some $\frac{3}{2}\leq\rho<2$; then, there exists a $\displaystyle \left(\frac{\displaystyle \rho}{\rho-e^{1-\frac{2}{\rho}}}\right)$-approximation algorithm for the metric prize-collecting \st path problem.
\end{thm}
\begin{proof}
The given algorithm is a polynomial-time algorithm. Let $P$ denote the output path.

It can be easily verified that $0<a<1$ and $0<p<1$. From Lemma~\ref{l:prounding},\begin{eqnarray}
\E[c(P)|\mathscr{A}^\rho\textrm{ is chosen}] &\leq& \E[\frac{\rho}{\gamma}c(x^*)|\mathscr{A}^\rho\textrm{ is chosen}]\nonumber\\
&=&\rho c(x^*)\int_{a}^{1}\frac{1}{1-a}\frac{1}{\gamma}d\gamma\nonumber\\
&=&\frac{-\ln a}{1-a}\rho c(x^*)\label{e:pcfinal1}
.\end{eqnarray}
We have\begin{eqnarray}
\E[\pi (V\setminus V(P))|\mathscr{A}^\rho\textrm{ is chosen}] &=& \sum_{v\in V\setminus\{s,t\}} \pi(v)\cdot\Pr[v\notin V_\gamma]\nonumber\\
&=&\sum_{v\in V\setminus\{s,t\}} \pi(v)\cdot\min\left(\frac{1-y^*_v}{1-a},1\right)\nonumber\\
&\leq& \frac{1}{1-a}\pi(\mathbf{1}-y^*)\label{e:pcfinal2}
.\end{eqnarray}

From \eqref{e:pcfinal1}, \eqref{e:pcfinal2}, and Lemma~\ref{l:cgrtpg},\begin{eqnarray*}
\E[c(P)+\pi(V\setminus V(P))] &=& p \left[ 2c(x^*)+\pi(\mathbf{1}-y^*) \right] + (1-p) \left[ \frac{-\ln a}{1-a}\rho c(x^*) + \frac{1}{1-a}\pi(\mathbf{1}-y^*) \right]\\
&=& \left[2p+(1-p)\frac{-\ln a}{1-a}\rho\right] c(x^*) + \left[p+ (1-p)\frac{1}{1-a} \right] \pi(\mathbf{1}-y^*)\\
&=& \frac{\rho}{\rho-e^{1-\frac{2}{\rho}}}\left[ c(x^*)+\pi(\mathbf{1}-y^*) \right]
.\end{eqnarray*}
\end{proof}

Theorem~\ref{t:pcfinal} along with Theorem~\ref{t:main} yields the following:
\begin{cor}
There exists a deterministic $1.9535$-approximation algorithm for the metric prize-collecting \st path problem.
\end{cor}

\begin{cor}
The integrality gap of \eqref{e:pca} is smaller than $1.9535$.
\end{cor}

\subsection{Unit-weight graphical metrics}\label{ss:unit}

In this subsection, we study the \st path TSP for a special case where the cost function is a shortest-path metric defined by an underlying undirected, unit-weight graph.

Let $x^*$ be an optimal solution to the path-variant Held-Karp relaxation; $G_0$ be the underlying unit-weight graph defining the cost function. $G_0$ is connected.

Mucha~\cite{M} gives an improved analysis of the $1.5858$-approximation algorithm of M\"omke and Svensson~\cite{MS}; following is from \cite{M}.

\begin{lemma}[\cite{M}]\label{l:m}
There exists an algorithm $\mathscr{A}_{0}$ for the \st path TSP under unit-weight graphical metrics, which returns a solution of cost at most\[
\min\left(\frac{10}{9}c(x^*)+\frac{1}{3}c(s,t)+\frac{1}{3}|V|+\frac{4}{9} , 2|V|-2-c(s,t) \right)
.\]
\end{lemma}

This immediately gives a $(\frac{19}{12}+\epsilon)$-approximation algorithm for any $\epsilon>0$. $\frac{19}{12}<1.5834$.
\begin{thm}[\cite{M}]\label{t:m}
There exists a $(\frac{19}{12}+\epsilon)$-approximation algorithm for the \st path TSP under unit-weight graphical metrics, for any $\epsilon>0$.
\end{thm}
\begin{proof}
Let $P$ be the output of $\mathscr{A}_{0}$. From Lemma~\ref{l:m},\begin{eqnarray*}
c(P) &\leq& \frac{3}{4} \left( \frac{10}{9}c(x^*)+\frac{1}{3}c(s,t)+\frac{1}{3}|V|+\frac{4}{9} \right) + \frac{1}{4} \left( 2|V|-2-c(s,t) \right)\\
&=& \frac{5}{6}c(x^*)+\frac{3}{4}(|V|-1)+\frac{7}{12}\\
&\leq& \frac{5}{6}c(x^*)+\frac{3}{4}c(x^*)+\frac{7}{12},
\end{eqnarray*}where the last line holds since $c(e)\geq 1$ for all $e$.

Thus, there exists $n_0$ such that $c(P)\leq (\frac{19}{12}+\epsilon)c(x^*)$ for all input that has $n_0$ or more vertices. Smaller instances can be separately solved.
\end{proof}

It can be observed from Lemma~\ref{l:m} and Theorem~\ref{t:m} that the ``critical case'' determining the proven performance guarantee is when $c(x^*)\approx |V|$. We will show that three different constructions of Hamiltonian paths carry performance analyses with complementary critical cases.

Even though $\tau$-narrow cuts function as a mere analytic tool in Section~\ref{s:improv}, we propose an algorithm that actually computes the $\tau$-narrow cuts and utilize them: once the $\tau$-narrow cuts are computed, the algorithm constructs an \st path that traverses from the first layer to the last, without ``skipping'' any layer in-between. If the path is inexpensive, the number of $\tau$-narrow cuts is also small so the algorithm presented in Section~\ref{s:improv} produces a good solution. If the path is expensive but the Held-Karp optimum is close to $|V|-1$, then we prove that the path already contains a large number of vertices and therefore can be augmented into a spanning Eulerian path with small additional cost. Lastly, if the Held-Karp optimum is bounded away from $|V|-1$, then M\"omke \& Svensson's algorithm performs well provided that the graph has large number of vertices.

Algorithm~\ref{a:unitg} shows the entire algorithm (except the separate handling of small instances); $\theta\in(0,1)$ is a parameter to be chosen later. Let $\eta:E\to\mathbb{Z}_{\geq 0}$ be a function such that $\eta(e):=c(e)-1$. For $U\subset V$, $G(U)$ denotes the subgraph of $G$ induced by $U$. Suppose $|V|\geq 3$; this implies $\ell\geq 3$.

\begin{algorithm}[ht]
\caption{The algorithm for the \st path TSP under unit-weight graphical metrics}
\label{a:unitg}
\begin{algorithmic}[1]
	\REQUIRE Complete graph $G=(V,E)$ with cost function $c:E\to\mathbb{Z}_{>0}$; endpoints $s,t\in V$.
	\ENSURE Hamiltonian path between $s$ and $t$.
	\STATE Run $\mathscr{A}_{0}$; let $H_A$ be the output Hamiltonian path.
	\STATE $x^*\gets$an optimal solution to the path-variant Held-Karp relaxation
	\STATE Run the algorithm from Section~\ref{s:improv}; let $H_B$ be the output Hamiltonian path.
	\STATE Compute the partition $L_1,\ldots L_\ell$ defining all the $(1-\theta)$-narrow cuts $U_i$.\label{as:unitg:0}
	\FOR {$1\leq i <\ell$}
	\STATE Let $(p_i,q_{i+1})$ be the shortest edge in $E(L_i,L_{i+1})$, where $p_i\in L_i$ and $q_{i+1}\in L_{i+1}$.
	\ENDFOR
	\FOR {$1< i <\ell$}
	\STATE Let $P_i$ be the shortest path from $q_i$ to $p_i$ within $G(L_i)$, under edge cost given by $\eta$.
	\ENDFOR
	\STATE Let $P_{\mathsf{LT}}$ be an \st path obtained by concatenating $(s,q_2),P_2,(p_2,q_3),\ldots,P_{\ell-1},(p_{\ell-1},t)$.\hspace{-2em}
	\STATE $G_E\gets (V,P_{\mathsf{LT}})$
	\WHILE {the multigraph $G_E$ is not spanning}\label{as:unitg:1}
	\STATE Choose $(u,v)$ such that: $c(u,v)=1$, $u$ is isolated in $G_E$, and $v$ is not.
	\STATE Add two copies of $(u,v)$ to $G_E$.
	\ENDWHILE \label{as:unitg:2}
	\STATE Shortcut an Eulerian path of $G_E$ to obtain a Hamiltonian path $H_C$.
	\STATE Let $H_{\mathsf{out}}$ be the best among $H_A$, $H_B$ and $H_C$; output $H_{\mathsf{out}}$.
\end{algorithmic}
\end{algorithm}

\begin{lemma}\label{l:wd}
Algorithm~\ref{a:unitg} is a well-defined, polynomial-time algorithm.
\end{lemma}
\begin{proof}
Steps~\ref{as:unitg:1}-\ref{as:unitg:2} start with an \st path, and augment it into a spanning multigraph that has an Eulerian path between $s$ and $t$. This follows from the preservation of the parity of degree. Choice of $(u,v)$ satisfying $c(u,v)=1$ is always possible since $G_0$ is connected.

$P_{\mathsf{LT}}$ is an \st path since $L_1=\{s\}$ and $L_\ell=\{t\}$. Note that some of $P_i$'s may be a length-0 path.

Step~\ref{as:unitg:0}, unlike the algorithm from Section~\ref{s:improv}, actually computes the layered structure of $(1-\theta)$-narrow cuts, whereas this structure was only for the sake of analysis in Section~\ref{s:improv}. Yet, the layers can in fact be identified via a polynomial number of min-cut calculations; hence, the algorithm is a polynomial-time algorithm.
\end{proof}

\begin{lemma}\label{l:c1}
\[
x^*(E(L_1,L_2))>\theta
.\]
\end{lemma}
\begin{proof}
We have\begin{equation}\label{e:c1:1}
x^*(E(L_1,L_{\geq 3}))+x^*(E(L_2,L_{\geq 3})) = x^*(\delta(U_2))<2-\theta
\end{equation}and\begin{equation}\label{e:c1:2}
x^*(E(L_1,L_2)+x^*(E(L_2,L_{\geq 3})) = x^*(\delta(L_2))\geq 2
.\end{equation}From \eqref{e:c1:1} and \eqref{e:c1:2},\[
x^*(E(L_1,L_2))-x^*(E(L_1,L_{\geq 3})) >\theta
.\]
\end{proof}
By symmetry, $x^*(E(L_{\ell-1},L_\ell))>\theta$.

\begin{lemma}\label{l:c2}
For any $i\geq 1$, $j\leq\ell$, $V_1\neq\emptyset$ and $V_2\neq\emptyset$ such that\begin{enumerate}
\item $i+2\leq j$,
\item $V_1\cup V_2=\cup_{k=i+1}^{j-1} L_k$, and
\item $V_1\cap V_2 =\emptyset$,
\end{enumerate}then $x^*(E(V_1,V_2))>\theta$.
\end{lemma}
\begin{proof}
We have\begin{equation}\label{e:c2:1}
x^*(E(L_{\leq i},V_1))+x^*(E(L_{\leq i},V_2))+x^*(E(L_{\leq i},L_{\geq j})) = x^*(\delta(L_{\leq i})) <2-\theta
;\end{equation}by symmetry,\begin{equation}\label{e:c2:2}
x^*(E(L_{\leq i},L_{\geq j}))+x^*(E(V_1,L_{\geq j}))+x^*(E(V_2,L_{\geq j})) <2-\theta
;\end{equation}\begin{equation}\label{e:c2:3}
x^*(E(L_{\leq i},V_1))+x^*(E(V_1,V_2))+x^*(E(V_1,L_{\geq j})) = x^*(\delta(V_1)) \geq 2
;\end{equation}again by symmetry,\begin{equation}\label{e:c2:4}
x^*(E(L_{\leq i},V_2))+x^*(E(V_1,V_2))+x^*(E(V_2,L_{\geq j})) \geq 2
.\end{equation}

From \eqref{e:c2:1} through \eqref{e:c2:4},\[
2x^*(E(V_1,V_2))-2x^*(E(L_{\leq i},L_{\geq j}))>2\theta
.\]
\end{proof}

\begin{cor}\label{c:ccut}
For all $1\leq i<\ell$, $x^*(E(L_i,L_{i+1}))>\theta$.
\end{cor}
\begin{proof}
From Lemma~\ref{l:c1} and Lemma~\ref{l:c2} applied for $j-i=3$.
\end{proof}

\begin{cor}\label{c:clayer}
For all $i$, $G(L_i)$ weighted by (the projection of) $x^*$ is $\theta$-edge-connected.
\end{cor}
\begin{proof}
$L_1$ and $L_\ell$ are singleton; every cut in any other nonsingleton layer subgraphs are of capacity at least $\theta$ from Lemma~\ref{l:c2}, applied for $j-i=2$.
\end{proof}

Let $\sigma,\kappa\geq 0$ be some parameters to be chosen later.

\begin{lemma}\label{l:unitgmain}
\[
c(H_{\mathsf{out}})\leq \max\left\{
\begin{array}{l}
\displaystyle \left( \frac{5}{6}+\frac{3}{4(1+\sigma)} \right) c(x^*)+\frac{7}{12}\\
\\
\displaystyle \left( 2-\kappa+\frac{2\sigma}{\theta} \right) c(x^*)\\
\\
\displaystyle \left[ \frac{3+2\theta}{2+\theta} + \frac{(1-\theta)^2}{4(2+\theta)} \kappa \right] c(x^*)
\end{array}\right\}
.\]
\end{lemma}
\begin{proof}
Suppose $c(x^*)\geq (1+\sigma)(|V|-1)$; from the proof of Theorem~\ref{t:m},\begin{eqnarray*}
c(H_{\mathsf{out}}) &\leq& c(H_A)\\
&\leq& \frac{5}{6}c(x^*)+\frac{3}{4}(|V|-1)+\frac{7}{12}\\
&\leq& \left(\frac{5}{6}+\frac{3}{4(1+\sigma)}\right)c(x^*)+\frac{7}{12}
;\end{eqnarray*}thus, we can assume from now that $c(x^*)<(1+\sigma)(|V|-1)$.

\paragraph{Case 1.}\begin{equation}\label{e:case1}
c(P_{\mathsf{LT}})\geq \kappa (|V|-1)
.\end{equation}

From Corollary~\ref{c:ccut} and the choice of $(p_i,q_{i+1})$,\begin{equation}\label{e:ugm11}
\theta\cdot\eta(p_i,q_{i+1})\leq (\eta\ast x^*)(E(L_i,L_{i+1}))
.\end{equation}For each layer $L_i$ with $1<i<\ell$, consider a bidirected flow network on $G(L_i)$ whose capacities are given by $x^*$. From Corollary~\ref{c:clayer}, we can route flow of $\theta$ from $q_i$ to $p_i$. This flow can be decomposed into cycles and paths from $q_i$ to $p_i$; thus, by the choice of $P_i$,\begin{equation}\label{e:ugm12}
\theta\cdot\eta(P_i)\leq (\eta\ast x^*)(E(L_i))
.\end{equation}

From \eqref{e:ugm11} and \eqref{e:ugm12},\begin{eqnarray}
\theta\cdot\eta(P_{\mathsf{LT}}) &=& \sum_{1\leq i<\ell} \theta\cdot \eta(p_i,q_{i+1}) + \sum_{1< i<\ell} \theta\cdot \eta(P_i)\nonumber\\
&\leq& \sum_{1\leq i<\ell} (\eta\ast x^*)(E(L_i,L_{i+1})) + \sum_{1< i<\ell} (\eta\ast x^*)(E(L_i))\nonumber\\
&\leq& (\eta\ast x^*)(E)\nonumber\\
&=&c(x^*)-x^*(E)\nonumber\\
&=&c(x^*)-(|V|-1)\nonumber\\
&<&\sigma(|V|-1)\label{e:unitg1f}
.\end{eqnarray}

Let $G'_E$ be $G_E$ after finishing the execution of Steps~\ref{as:unitg:1}-\ref{as:unitg:2} of Algorithm~\ref{a:unitg}; $|P_{\mathsf{LT}}|$ denotes the number of edges on $P_{\mathsf{LT}}$. We have\begin{eqnarray*}
c(H_{\mathsf{out}}) &\leq& c(H_C)\\
&\leq& c(G'_E)\\
&=& c(P_{\mathsf{LT}}) + 2\left[ (|V|-1)-|P_{\mathsf{LT}}| \right]\\
&=& c(P_{\mathsf{LT}}) + 2\left[ (|V|-1)-\{c(P_{\mathsf{LT}})-\eta(P_{\mathsf{LT}})\} \right]\\
&=& 2(|V|-1) -c(P_{\mathsf{LT}}) + 2\eta(P_{\mathsf{LT}})\\
&\leq& \left[ 2- \kappa +\frac{2\sigma}{\theta}\right] \cdot (|V|-1)\\
&\leq& \left( 2-\kappa+\frac{2\sigma}{\theta} \right) c(x^*)
,\end{eqnarray*}where the second last line follows from \eqref{e:case1} and \eqref{e:unitg1f}; the last from $c(x^*)\geq |V|-1$.

\paragraph{Case 2.}\begin{equation}\label{e:case2}
c(P_{\mathsf{LT}}) < \kappa (|V|-1)
.\end{equation}Note that, from the construction of $P_{\mathsf{LT}}$, $\ell-1\leq |P_{\mathsf{LT}}|$; hence we have\[
\ell-1\leq |P_{\mathsf{LT}}| \leq c(P_{\mathsf{LT}}) < \kappa (|V|-1)
.\]

From each $(1-\theta)$-narrow cut $(U_i,\bar U_i)$, we can pick an edge $d_i\in\delta(U_i)$ with $c(d_i)=1$ due to the connectedness of $G_0$. Let $\hat f^*_{U_i} := \mathbf{e}_{d_i}$, $\alpha:=\frac{\theta}{2+\theta}$, $\beta:=\frac{1}{2+\theta}$, and $\tau=\frac{1-2\alpha}{\beta}-1=1-\theta$. Note that this choice of $\alpha$ and $\beta$ satisfies \eqref{e:later1}. Since the second condition on $\{\hat f _{U_i}^*\}_{i=1}^{\ell-1}$ of Lemma~\ref{l:fd} is not used to derive \eqref{e:later2} (it is used in the later part of the proof), we have\begin{eqnarray*}
c(H_{\mathsf{out}}) &\leq& c(H_B)\\
&\leq& \E[c(H)]\\
&\leq& (1+\alpha+\beta)c(x^*)+ \left\{\max_{0\leq\omega\leq \tau} \omega\left[1-\{2\alpha+\beta(1+\omega)\}\right] \right\} c\left(\sum_{i=1}^{\ell-1} \hat f ^*_{U_i}\right)\\
&=& \frac{3+2\theta}{2+\theta}c(x^*)+ \frac{(1-\theta)^2}{4(2+\theta)} c\left(\sum_{i=1}^{\ell-1} \hat f ^*_{U_i}\right)
.\end{eqnarray*}As $c(d_i)=1$ for all $i$,\begin{eqnarray*}
c(H_{\mathsf{out}}) &\leq& \frac{3+2\theta}{2+\theta}c(x^*)+ \frac{(1-\theta)^2}{4(2+\theta)} (\ell -1)\\
&\leq& \frac{3+2\theta}{2+\theta}c(x^*)+ \frac{(1-\theta)^2}{4(2+\theta)} \kappa (|V|-1)\\
&\leq& \left\{ \frac{3+2\theta}{2+\theta} + \frac{(1-\theta)^2}{4(2+\theta)} \kappa \right\} c(x^*)
.\end{eqnarray*}
\end{proof}

\begin{cor}\label{c:unitgar}
Let $\rho:=\max\left\{
\frac{5}{6}+\frac{3}{4(1+\sigma)},
2-\kappa+\frac{2\sigma}{\theta},
\frac{3+2\theta}{2+\theta} + \frac{(1-\theta)^2}{4(2+\theta)} \kappa
\right\}$. There exists a $(\rho+\epsilon)$-ap\-proxi\-ma\-tion algorithm for the \st path TSP under unit-weight graphical metrics, for any $\epsilon>0$.
\end{cor}

\begin{cor}
There exists a $1.5780$-approximation algorithm for the \st path TSP under unit-weight graphical metrics.
\end{cor}
\begin{proof}
Directly follows from Corollary~\ref{c:unitgar}: if we choose, for example, $\theta=1.2297\times 10^{-1}$, $\sigma=7.2774\times 10^{-3}$, and $\kappa=5.4045\times 10^{-1}$, we have $\rho<1.5780$.
\end{proof}

\begin{cor}
The integrality gap of the path-variant Held-Karp relaxation under the unit-weight graphical metric is smaller than $1.6137$.
\end{cor}
\begin{proof}
Trivial for $|V|=2$. Let $\mathsf{OPT}$ denote the optimal (integral) solution value.

Suppose $3\leq |V|\leq 6$. From a similar argument as in the proof of Lemma~\ref{l:wd}, if there exists a simple \st path with $m$ edges in $G_0$, $\mathsf{OPT}\leq m+2(|V|-1-m)=2|V|-2-m$. Thus, if there exists a simple \st path with at least two edges,\[
\frac{\mathsf{OPT}}{c(x^*)}\leq\frac{2|V|-4}{|V|-1}\leq\frac{8}{5} <1.6137
.\]Suppose there does not exist a simple \st path with more than one edge; then $(s,t)\in G_0$ and $(s,t)$ is a bridge of $G_0$. Let $(U,\bar U)$ be the \st cut defined by the removal of $(s,t)$ from $G_0$. $x^*(s,t)=0$ since $2x^*(s,t)=x^*(\delta(\{s\}))+ x^*(\delta(\{t\})) -x^*(\delta(\{s,t\}))\leq 1+1-2=0$; therefore,\begin{eqnarray*}
c(x^*)&=&(c\ast x^*)(\delta(U))+(c\ast x^*)(E\setminus\delta(U))\\
&=&(c\ast x^*)(\delta(U)\setminus\{s,t\})+(c\ast x^*)(E\setminus\delta(U))\\
&\geq&2x^*(\delta(U)\setminus\{s,t\})+x^*(E\setminus\delta(U))\\
&=&x^*(\delta(U))+x^*(E)\\
&\geq&|V|
\end{eqnarray*}and\[
\frac{\mathsf{OPT}}{c(x^*)}\leq\frac{2|V|-3}{|V|}\leq\frac{3}{2} <1.6137
.\]

Suppose $|V|\geq 7$. Choose $\theta=3.7304\times 10^{-1}$, $\sigma=8.5757\times 10^{-2}$, and $\kappa=8.4614\times 10^{-1}$; from Lemma~\ref{l:unitgmain},\begin{eqnarray*}
c(H_{\mathsf{out}})&\leq& \max\left\{
\begin{array}{l}
\displaystyle \left( \frac{5}{6}+\frac{3}{4(1+\sigma)} +\frac{7}{12(|V|-1)(1+\sigma)} \right) c(x^*)\\
\\
\displaystyle \left( 2-\kappa+\frac{2\sigma}{\theta} \right) c(x^*)\\
\\
\displaystyle \left[ \frac{3+2\theta}{2+\theta} + \frac{(1-\theta)^2}{4(2+\theta)} \kappa \right] c(x^*)
\end{array}
\right\}\\
&<&Qc(x^*)
,\end{eqnarray*}for some $Q<1.6137$.
\end{proof}

\section{Open questions}\label{s:oq}

An immediate open question is in improving the performance guarantee. The fractional $T$-join dominators constructed in the analyses are not directly derived from the algorithm; a different construction may lead to an improved performance guarantee. One related question is whether $\alpha$ and $\beta$ can be chosen differently. In the proof of $\left(\frac{1+\sqrt{5}}{2}\right)$-approximation, Lemma~\ref{l:fd} can be considered as distributing $c(x^*)$ over the cuts of different capacities. An adaptive choice of $\alpha$ and $\beta$ after seeing one such distribution does not appear to improve the analysis; from Yao's Lemma, oblivious but stochastic choice of $\alpha$ and $\beta$ does not either.

A bigger open question is whether the techniques presented in this paper can be extended to the circuit case as well. Given the successful adaptation of the techniques devised in one variant to the other in the unit-weight graphical metric case, whether the present techniques can be extended to beat the longstanding $3/2$ barrier of the general-metric circuit problem becomes an interesting question. It appears that the layered structure of $\tau$-narrow cuts or the parity argument on them are less likely to directly extend to the circuit case, as the arguments rely on the characteristics of the path case; what could be more promising is the approach of repairing deficient cuts using a set of vectors obtained from an auxiliary flow network, since this approach might extend to work with some different type of ``fragile cut structure''.

\bibliography{path2}

\appendix

\section{An LP-based new analysis of the path-variant Christofides' algorithm}\label{ap:c53}

In this appendix, we present a new analysis of the path-variant Christofides' algorithm~\cite{C, H} for the metric \st path TSP, and show how the critical case characterized by this analysis can lead to an improvement. The analysis compares the output solution value to the LP optimum of the path-variant Held-Karp relaxation, thereby proving the upper bound of $5/3$ on the integrality gap of the path-variant Held-Karp relaxation. We note that the LP optimum is never computed by the algorithm.

First we recall the following definition of the circuit-variant Held-Karp relaxation:
\begin{defn}[\cite{HK}]
\label{d:hkcircuit}
The \emph{circuit-variant Held-Karp relaxation} is the following:\begin{equation}\label{e:hkcircuit}
\begin{array}{lll}
\textrm{minimize}&c(x)&\\
\textrm{subject to}&x(\delta(S))\geq 2,&\forall S\subsetneq V, S\neq\emptyset;\\
&x(\delta(\{v\})) = 2,&\forall v\in V;\\
&x\geq 0.&
\end{array}
\end{equation}
\end{defn}

Let $G=(V,E)$ be the input complete graph with cost function $c:E\to\mathbb{R}_+$ and the endpoints $s,t\in V$. The path-variant Christofides' algorithm first finds a minimum spanning tree $\mathscr{T}_{\mathrm{min}}$ of $G$; it then computes a minimum $T$-join $J$, where $T\subset V$ is the set of the vertices with the ``wrong'' parity of degree in $\mathscr{T}_{\mathrm{min}}$: i.e., $T$ is the set of odd-degree internal points and even-degree endpoints in $\mathscr{T}_{\mathrm{min}}$. Lastly, the algorithm shortcuts an Eulerian path of the multigraph $\mathscr{T}_{\mathrm{min}}\cup J$ to obtain the output Hamiltonian path $H$.

We give two different bounds on the cost of $J$, which together will establish the performance guarantee. Let $x^*\in\mathbb{R}^E$ be the LP optimum of the path-variant Held-Karp relaxation.

\begin{lemma}
\label{l:c53j1}
$c(\mathscr{T}_{\mathrm{min}})\leq c(x^*)$.
\end{lemma}
\begin{proof}
As can be seen from Observation~\ref{o:hkequiv}, the path-variant Held-Karp polytope is contained in the spanning tree polytope. The lemma follows from this observation, since $\mathscr{T}_{\mathrm{min}}$ is a minimum spanning tree.
\end{proof}

Lemmas~\ref{l:c53j2} and \ref{l:c53j3} give the two bounds.

\begin{lemma}
\label{l:c53j2}
$c(J)\leq \frac{1}{2}\left\{c(x^*)+c(s,t)\right\}$.
\end{lemma}
\begin{proof}
Let $x^*_{\mathrm{circuit}} := x^*+\mathbf{e}_{(s,t)}$: i.e., $x^*_{\mathrm{circuit}}$ is obtained by ``adding'' the edge $(s,t)$ to $x^*$. Then $x^*_{\mathrm{circuit}}$ is a feasible solution to the circuit-variant Held-Karp relaxation (see \eqref{e:hkpath1} and \eqref{e:hkcircuit}). Let $\mathrm{HK_{circuit}}$ be the optimal value of the circuit-variant Held-Karp relaxation and we have\begin{eqnarray*}
c(J)&\leq& \frac{1}{2}\mathrm{HK_{circuit}}\\
& \leq& \frac{1}{2} c(x^*_{\mathrm{circuit}})\\
& =& \frac{1}{2} \left\{ c(x^*)+c(s,t) \right\}
,\end{eqnarray*}where the first inequality follows from \cite{W, SW}.
\end{proof}

\begin{lemma}
\label{l:c53j3}
$c(J)\leq c(x^*)-c(s,t)$.
\end{lemma}
\begin{proof}
Let $P_{st}^{\mathscr{T}_{\mathrm{min}}}$ be the path between $s$ and $t$ on $\mathscr{T}_{\mathrm{min}}$. Consider an edge set $J':=\mathscr{T}_{\mathrm{min}}\setminus P_{st}^{\mathscr{T}_{\mathrm{min}}}$. Note that $J'$ is a $T$-join: $v\in V$ has even degree in $P_{st}^{\mathscr{T}_{\mathrm{min}}}$ if and only if $v$ is internal; thus, $v$ has even degree in the multigraph $\mathscr{T}_{\mathrm{min}}\cup J'=(\mathscr{T}_{\mathrm{min}}\cup \mathscr{T}_{\mathrm{min}})\setminus P_{st}^{\mathscr{T}_{\mathrm{min}}}$ if and only if $v$ is an internal point, and this shows that $v$ has odd degree in $J'$ if and only if $v\in T$.

We have\begin{eqnarray*}
c(J) &\leq& c(J')\\
&=& c(\mathscr{T}_{\mathrm{min}})-c(P_{st}^{\mathscr{T}_{\mathrm{min}}})\\
&\leq& c(x^*) -c(s,t)
.\end{eqnarray*}The last inequality follows from Lemma~\ref{l:c53j1} and the triangle inequality.
\end{proof}

\begin{thm}
$c(H)\leq\frac{5}{3}c(x^*)$; therefore, the path-variant Christofides' algorithm is a $5/3$-approximation algorithm, and the integrality gap of the path-variant Held-Karp relaxation is at most $5/3$.
\end{thm}
\begin{proof}
We have
\begin{eqnarray}
c(H) &\leq& c(\mathscr{T}_{\mathrm{min}})+c(J)\nonumber\\
&\leq& c(x^*)+\min\left [\frac{1}{2}\left\{c(x^*)+c(s,t)\right\} , c(x^*)-c(s,t) \right ]\nonumber\\
&=& \frac{5}{3}c(x^*)+\min\left [\frac{1}{2}\left\{-\frac{1}{3}c(x^*)+c(s,t)\right\} , \frac{1}{3}c(x^*)-c(s,t) \right ]\nonumber\\
&\leq& \frac{5}{3}c(x^*)\label{e:crit}
,\end{eqnarray}where the second inequality follows from Lemmas~\ref{l:c53j1}, \ref{l:c53j2} and \ref{l:c53j3}.
\end{proof}

We observe that the equality of \eqref{e:crit} is achieved when $c(s,t)=\frac{1}{3}c(x^*)$, and this is the critical case of this analysis that determines the performance guarantee proven. Hence, if we can improve the performance guarantee only near this critical case, such an improvement would lead to a better approximation ratio. We demonstrate this approach, by presenting how this analysis combines with the results of Oveis Gharan et al.~\cite{OSS} on the unit-weight graphical metric TSP to yield a comparable result in the \st path TSP.

We consider the \st path TSP under the unit-weight graphical metric; we show how to modify the algorithm of Oveis Gharan et al. for the path case and that, when $c(s,t)$ is close to $\frac{1}{3}c(x^*)$, this modified algorithm carries a performance guarantee that is slightly better than $5/3$.

First we review the results in Oveis Gharan et al.~\cite{OSS}. In the following, the parameters $\epsilon_1, \epsilon_2, \gamma, \delta$ and $\rho$ can be chosen as follows: $\epsilon_1 = 1.875\cdot 10^{-12}$, $\epsilon_2 = 5\cdot 10^{-2}$, $\gamma = 10^{-7}$, $\delta=6.25 \cdot 10^{-16}$, $\rho = 1.5\cdot 10^{-24}$, and $n$ denotes $|V|$.

\begin{defn}[Nearly integral edges]\label{d:niedges}
An edge $e$ is \emph{nearly integral} with respect to $x\in\mathbb{R}^E$ if $x_e\geq 1-\gamma$.
\end{defn}

\begin{defn}\label{d:approxm}
For some constant $\nu\leq\frac{1}{5}$ and $k\geq 2$, a \emph{maximum entropy distribution over spanning trees with approximate marginal} $x\in\mathbb{R}^E$ is a probability distribution $\mu$ defined by $\lambda\in\mathbb{R}^E$ such that $\mu(\mathscr{T})\propto\prod_{e\in T}\lambda_e$ for every spanning tree $\mathscr{T}$ and the marginal probability of every edge $e$ is no greater than $(1+\frac{\nu}{n^k})x_e$.
\end{defn}

\begin{defn}[Good edges]\label{d:goodedges}
A cut is $(1+\delta)$-near-minimum if its weight is at most $(1+\delta)$ times the minimum cut weight. An edge $e$ is \emph{even} with respect to $F\subset E$ if every $(1+\delta)$-near-minimum cut containing $e$ has even number of edges intersecting with $F$.

For a circuit-variant Held-Karp feasible solution $x^*_{\mathrm{circuit}}$, consider $x^*_{\mathrm{circuit}}$ as the edge weight and let $F$ be a spanning tree sampled from a maximum entropy distribution with approximate marginal $(1-\frac{1}{n}) x^*_{\mathrm{circuit}}$. We say an edge $e$ is \emph{good} with respect to $x^*_{\mathrm{circuit}}$ if the probability that $e$ is even with respect to $F$ is at least $\rho$.
\end{defn}

\begin{thm}[Structure Theorem]
\label{t:structure}
Let $x^*_{\mathrm{circuit}}$ be a feasible solution to the circuit-variant Held-Karp relaxation, and let $\mu$ be a maximum entropy distribution over spanning trees with approximate marginal $(1-\frac{1}{n}) x^*_{\mathrm{circuit}}$. There exist small constants $\epsilon_1 , \epsilon_2 >0$ such that at least one of the following is true:\begin{itemize}
\item[1.] there exists a set $E^*\subset E$ such that $x(E^*)\geq \epsilon_1 n$ and every edge in $E^*$ is good with respect to $x^*_{\mathrm{circuit}}$;
\item[2.] there exist at least $(1-\epsilon_2)n$ edges that are nearly integral with respect to $x^*_{\mathrm{circuit}}$.
\end{itemize}
\end{thm}

\begin{lemma}
\label{l:case1}
Suppose that Case~1 of Theorem~\ref{t:structure} holds and $\mathscr{T}$ is sampled from $\mu$. Let $T$ be the set of odd-degree vertices in $\mathscr{T}$, then a minimum $T$-join $J$ satisfies\[
\E[c(J)]\leq c(x^*_{\mathrm{circuit}}) (\frac{1}{2}-\frac{\epsilon_1\delta\rho}{4(1+\delta)})
.\]
\end{lemma}

We are now ready to present the algorithm. Algorithm~\ref{a:au} describes the entire algorithm for the \st path TSP under the unit-weight graphical metric. It first computes the LP optimum $x^*$. If $c(s,t)$ is close to $\frac{1}{3}c(x^*)$, we run a modified version of Oveis Gharan, Saberi, and Singh's algorithm (Cases~A1 and A2); otherwise, we invoke Christofides' algorithm (Case~B). Parameters $\sigma_l,\sigma_u$ and $\epsilon'_2$ are to be chosen later.

\begin{algorithm}[ht]
\caption{Algorithm for the \st path TSP under the unit-weight graphical metric}
\label{a:au}
\begin{algorithmic}[1]
	\REQUIRE Complete graph $G=(V,E)$ with cost function $c:E\to\mathbb{Z}_{>0}$; endpoints $s,t\in V$.
	\ENSURE Hamiltonian path between $s$ and $t$.

	\STATE $x^*\gets$optimal solution to the path-variant Held-Karp relaxation
	\IF {$c(s,t)=(\frac{1}{3}+\alpha)c(x^*)$ for $\alpha\in[-\sigma_l,\sigma_u]$}
		\IF[Case A1]{at least $(1-\epsilon'_2)(n-1)$ edges are nearly integral w.r.t. $x^*$}
			\STATE Find a minimum spanning subgraph $F'$ containing all the nearly integral edges.
			\STATE Find a minimum spanning tree $\mathscr{T}$ of $F'$.
			\STATE Let $T$ be the set of odd-degree internal points and even-degree endpoints in $\mathscr{T}$.
			\STATE Compute a minimum $T$-join $J$; $\mathscr{L} \gets \mathscr{T}\cup J$.
		\ELSE[Case A2]
			\STATE $x^*_{\mathrm{circuit}} := x^*+\mathbf{\mathbf{e}}_{(s,t)}$
			\STATE Sample spanning tree $\mathscr{T}$ from max-entropy distribution with approx. marginal $(1-\frac{1}{n}) x^*_{\mathrm{circuit}}$.
			\STATE Let $T$ be the set of odd-degree vertices in $\mathscr{T}$.
			\STATE Compute a minimum $T$-join $J$; $\mathscr{L}_0 \gets \mathscr{T}\cup J$.
			\STATE \textbf{if} $(s,t)\in \mathscr{L}_0$ \textbf{then} $\mathscr{L}\gets \mathscr{L}_0 \setminus\{(s,t)\}$ \textbf{else} $\mathscr{L}\gets \mathscr{L}_0 \cup\{(s,t)\}$ \textbf{end if}
		\ENDIF
	\ELSE[Case B]
		\STATE Find a minimum spanning tree $\mathscr{T}$ of $G$.
		\STATE Let $T$ be the set of odd-degree internal points and even-degree endpoints in $\mathscr{T}$.
		\STATE Compute a minimum $T$-join $J$; $\mathscr{L} \gets \mathscr{T}\cup J$.
	\ENDIF
	\STATE Shortcut an Eulerian path of the multigraph $\mathscr{L}$ to obtain a Hamiltonian path $H$; output it.
\end{algorithmic}
\end{algorithm}

First we show that we can have a Structure Theorem analogous to Theorem~\ref{t:structure} by adjusting $\epsilon_2$ and replacing $n$ with $(n-1)$ in Case~2. The following corollary states that either there are good edges of significant weight with respect to $x^*_{\mathrm{circuit}}$ or there are many nearly integral edges with respect to $x^*$.
\begin{cor}
\label{c:pathstructure}
Let $x^*$ be a feasible solution to the path-variant Held-Karp relaxation and  $x^*_{\mathrm{circuit}} := x^*+\mathbf{\mathbf{e}}_{(s,t)}$. Let $\mu$ be a maximum entropy distribution over spanning trees with approximate marginal $(1-\frac{1}{n}) x^*_{\mathrm{circuit}}$. There exist small constants $\epsilon_1 , \epsilon'_2 >0$ such that at least one of the following is true:\begin{itemize}
\item[1.] there exists a set $E^*\subset E$ such that $x(E^*)\geq \epsilon_1 n$ and every edge in $E^*$ is good with respect to $x^*_{\mathrm{circuit}}$;
\item[2.] there exist at least $(1-\epsilon'_2)(n-1)$ edges that are nearly integral with respect to $x^*$.
\end{itemize}
\end{cor}
\begin{proof}
By Theorem~\ref{t:structure}, at least one of the two cases of Theorem~\ref{t:structure} holds. Case~1 of Theorem~\ref{t:structure} and Case~1 of this corollary are identical, so consider when Case~2 of Theorem~\ref{t:structure} holds.

Recall that $\epsilon_2$ was chosen as $5\cdot 10^{-2}$; we choose $\epsilon'_2=6\cdot 10^{-2}$.

Suppose $n\leq 19$. $x^*_{\mathrm{circuit}}$ has at least $(1-\epsilon_2)n$ nearly integral edges; thus, $x^*$ has at least $\lceil (1-\epsilon_2)n\rceil - 1 = n-1\geq (1-\epsilon'_2)(n-1)$ nearly integral edges.

Suppose $n\geq 20$. $x^*$ has at least\begin{eqnarray*}
(1-\epsilon_2)n-1&=&(1-\epsilon_2)(n-1)-\epsilon_2\\
&\geq& (1-\frac{20}{19}\epsilon_2)(n-1)\\
&\geq& (1-\epsilon'_2)(n-1)
\end{eqnarray*}nearly integral edges.
\end{proof}

\begin{lemma}
\label{l:fca1}
In Case~A1, $c(H)\leq (\frac{5}{3}-C_{A1})c(x^*)$ for some $c_{A1}>0$.
\end{lemma}
\begin{proof}
The following proof is adapted from \cite{OSS} and modified for the path case.

Let $S'$ be the set of nearly integral edges. Since the metric is defined by an unweighted connected graph, $c(F')=c(S')+|F'\setminus S'|\leq\frac{(c\ast x^*)(S')}{1-\gamma}+|F'\setminus S'|$. From $\gamma<\frac{1}{3}$, we know that $S'$ is a union of disjoint cycles and paths and the lengths of cycles are at least $\frac{1}{\gamma}$. Thus, $|\mathscr{T}\cap S'|\geq (n-1)(1-\epsilon'_2)(1-\gamma)$ and $|\mathscr{T}\setminus S'|\leq (n-1)(\epsilon'_2+\gamma)\leq c(x^*)(\epsilon'_2+\gamma)$. Let $S=S'\cap \mathscr{T}$.

We construct a fractional $T$-join dominator $y$ as follows.\[
y_e = \begin{cases}
1&\textrm{if }e\in \mathscr{T}\setminus S\\
x_e^*&\textrm{if }e\in E\setminus \mathscr{T}\\
\frac{x_e^*}{2(1-\gamma)}&\textrm{if }e\in S
\end{cases}
\]We claim that $y$ is a fractional $T$-join dominator. Let $(U,\bar U)$ be any cut that has an odd number of vertices in $T$ on one side. If there exists an edge $e\in (\mathscr{T}\setminus S)\cap \delta(U)$, then $y(\delta(U))\geq y_e = 1$. So suppose from now on that $\delta(U)\cap \mathscr{T}\subset S$. Then $\delta(U)\cap S=\delta(U)\cap \mathscr{T}$.

If $U$ is nonseparating, $U$ contains odd number of odd-degree vertices, and thus $|\delta(U)\cap \mathscr{T}|$ is odd. We have $x^*(\delta(U))\geq 2$ from the Held-Karp formulation and thus\[
\begin{cases}
y(\delta(U))\geq x^*(\delta(U)\setminus \mathscr{T})\geq 1&\textrm{if }|\delta(U)\cap \mathscr{T}|=1\\
y(\delta(U))\geq y(\delta(U)\cap S)\geq 3\frac{1-\gamma}{2(1-\gamma)} > 1&\textrm{if }|\delta(U)\cap S|\geq 3
.\end{cases}
\]

If $(U,\bar U)$ is an \st cut, then $U$ contains even number of odd-degree vertices, and thus $|\delta(U)\cap \mathscr{T}|$ is even. We have $(\delta(U)\cap \mathscr{T})\neq\emptyset$ since $\mathscr{T}$ is connected and\[
y(\delta(U))\geq y(\delta(U)\cap S)\geq 2\frac{1-\gamma}{2(1-\gamma)}=1
.\]

Thus $y$ is a fractional $T$-join dominator. Now,\begin{eqnarray*}
c(H)&\leq&c(\mathscr{T})+c(y)\\
&\leq&\frac{(c\ast x^*)(S)}{1-\gamma}+c(\mathscr{T}\setminus S)+c(\mathscr{T}\setminus S)+(c\ast x^*)(E\setminus \mathscr{T})+\frac{(c\ast x^*)(S)}{2(1-\gamma)}\\
&\leq&\frac{3(c\ast x^*)(S)}{2(1-\gamma)}+2c(x^*)(\epsilon'_2+\gamma)+(c\ast x^*)(E\setminus S)\\
&\leq&c(x^*)(\frac{3}{2(1-\gamma)}+2\epsilon'_2+2\gamma)\\
&\leq&c(x^*)(\frac{5}{3}-C_{A1})
\end{eqnarray*}for some $C_{A1}>0$. For example, we can choose $c_{A1}=4\cdot 10^{-2}$.
\end{proof}

\begin{lemma}
\label{l:fca2}
In Case~A2, $\E[c(H)]\leq (\frac{5}{3}-C_{A2})c(x^*)$ for some $C_{A2}>0$.
\end{lemma}
\begin{proof}
First we have\begin{eqnarray*}
\E[c(\mathscr{T})] &\leq& c\left((1+\frac{\nu}{n^k})(1-\frac{1}{n})x^*_{\mathrm{circuit}}\right)\\
&\leq& (1+\frac{1}{5n^2})(1-\frac{1}{n})(\frac{4}{3}+\alpha)c(x^*)\\
&\leq& (1-\frac{4}{5n})(\frac{4}{3}+\alpha)c(x^*)
.\end{eqnarray*}From Lemma~\ref{l:case1},\[
\E[c(J)]\leq(\frac{4}{3}+\alpha)c(x^*)(\frac{1}{2}-\frac{\epsilon_1\delta\rho}{4(1+\delta)})
.\]

We have\begin{eqnarray*}
\Pr[(s,t)\in L_0] &\geq& \Pr[(s,t)\in \mathscr{T}]\\
&=& n-1-\E[|\mathscr{T}\setminus(s,t)|]\\
&\geq& n-1-(n-2+\frac{1}{n})(1+\frac{\nu}{n^k})\\
&\geq& n-1-(n-2+\frac{1}{n})(1+\frac{1}{5n^2})\\
&\geq& 1-\frac{7}{5n}
\end{eqnarray*}and hence\begin{eqnarray*}
\E[c(H)]&\leq&\E[c(\mathscr{T})]+\E[c(J)]-(1-\frac{7}{5n})c(s,t)+\frac{7}{5n}c(s,t)\\
&\leq&c(x^*)\left\{(1-\frac{4}{5n})(\frac{4}{3}+\alpha)+(\frac{4}{3}+\alpha)(\frac{1}{2}-\frac{\epsilon_1\delta\rho}{4(1+\delta)})\right.\\
&&\phantom{c(x^*)}\left.-(1-\frac{7}{5n})(\frac{1}{3}+\alpha)+\frac{7}{5n}(\frac{1}{3}+\alpha)\right\}\\
&=&c(x^*)\left\{(\frac{5}{3}-\frac{\epsilon_1\delta\rho}{3(1+\delta)})+\alpha(\frac{1}{2}-\frac{\epsilon_1\delta\rho}{4(1+\delta)}) -\frac{1}{n} (\frac{2}{15}-2\alpha) \right\}\\
&\leq&c(x^*)(\frac{5}{3}-C_{A2})
\end{eqnarray*}for some $C_{A2}>0$ by choosing sufficiently small $\sigma_l,\sigma_u>0$. For example, we can choose $\sigma_l=7.8\cdot 10^{-52}$, $\sigma_u=3.9\cdot 10^{-52}$ and $C_{A2}=3.9\cdot 10^{-52}$.
\end{proof}

\begin{lemma}
\label{l:fcb}
In Case~B, $c(H)\leq (\frac{5}{3}-C_B)c(x^*)$ for some $C_B>0$.
\end{lemma}
\begin{proof}
Suppose that $c(s,t)<(\frac{1}{3}-\sigma_l)c(x^*)$. From Lemmas~\ref{l:c53j1} and \ref{l:c53j2}, it follows that\begin{eqnarray*}
c(H)&\leq& c(\mathscr{T})+c(J)\\
&<& c(x^*)+\frac{1}{2} \left\{c(x^*)+ (\frac{1}{3}-\sigma_l)c(x^*) \right\}\\
&=& \left(\frac{5}{3}-\frac{\sigma_l}{2}\right) c(x^*)
.\end{eqnarray*}

Suppose $c(s,t)>(\frac{1}{3}+\sigma_u)c(x^*)$. From Lemmas~\ref{l:c53j1} and \ref{l:c53j3},\begin{eqnarray*}
c(H)&\leq& c(\mathscr{T})+c(J)\\
&<&c(x^*)+\left\{c(x^*) - (\frac{1}{3}+\sigma_u) c(x^*)\right\}\\
&=& \left(\frac{5}{3}-\sigma_u\right) c(x^*)
.\end{eqnarray*}

Now choose $C_B:=\min(\frac{\sigma_l}{2},\sigma_u)$.
\end{proof}

Lemmas~\ref{l:fca1}, \ref{l:fca2} and \ref{l:fcb} yield the following theorem.

\begin{thm}
For some $\epsilon >0$, Algorithm~\ref{a:au} is a $(\frac{5}{3}-\epsilon)$-approximation algorithm for the \st path TSP under the unit-weight graphical metric.
\end{thm}
\begin{proof}
In Cases~A1 and B, the multigraph $\mathscr{L}$ is the union of a spanning tree and a $T$-join where $T$ is the set of the vertices with the wrong parity of degree. Thus, $\mathscr{L}$ has an Eulerian path between the two endpoints.

In Case~A2, $\mathscr{L}_0$ is Eulerian and hence 2-edge-connected; $\mathscr{L}\supset \mathscr{L}_0 \setminus\{(s,t)\}$ is therefore connected and $\mathscr{L}$ has an Eulerian path between the two endpoints.

By choosing $\epsilon=\min \{C_{A1}, C_{A2}, C_B\}$, $\epsilon =3.9\cdot 10^{-52}$ for example, we have $\E[c(H)]\leq (\frac{5}{3}-\epsilon )c(x^*)$ from Lemmas~\ref{l:fca1}, \ref{l:fca2} and \ref{l:fcb}. Thus, Algorithm~\ref{a:au} is a $(\frac{5}{3}-\epsilon)$-approximation algorithm.
\end{proof}

\phantom{\cite{Full, CRW}}

\end{document}